\documentclass[journal,twocolumn]{IEEEtran}
\usepackage{amsmath,amssymb,url}
\newtheorem{theorem}{Theorem}
\newtheorem{lemma}{Lemma}
\newtheorem{defn}{Definition}
\newtheorem{prpn}{Proposition}

\newtheorem{example}{Example}

\def\f{{\mathbb{F}}}

\def\R{{\mathbb{R}}}

\providecommand{\norm}[1]{\lVert#1\rVert}
\begin{document}
\title{Low-delay, High-rate Non-square Complex Orthogonal Designs}

\author{Smarajit Das,~\IEEEmembership{Student Member,~IEEE} and
        B. Sundar Rajan,~\IEEEmembership{Senior Member,~IEEE}%
\thanks{This work was supported through grants to B.S.~Rajan; partly by the
IISc-DRDO program on Advanced Research in Mathematical Engineering,  partly
by the Council of Scientific \& Industrial Research (CSIR, India) Research
Grant (22(0365)/04/EMR-II) and partly by the INAE Chair Professorship grant. Part of the material in this paper was presented in IEEE 2008 International Symposium on Information Theory (ISIT-2008), July 6-11, Toronto, Canada.} 
\thanks{Smarajit Das is a Post-doctoral fellow in the School of Technology and Computer science, Tata Institute of Fundamental Research, Mumbai, and B. Sundar Rajan is with the Department of Electrical Communication Engineering, Indian Institute of Science, Bangalore-560012, India. emails:smarajitd@gmail.com, bsrajan@ece.iisc.ernet.in.}}

\maketitle
\begin{abstract}
The maximal rate of a non-square complex orthogonal design for $n$ transmit antennas is $\frac{1}{2}+\frac{1}{n}$ if $n$ is even and $\frac{1}{2}+\frac{1}{n+1}$ if $n$ is odd and 
the codes have been constructed for all $n$
by Liang (IEEE Trans. Inform. Theory, 2003) and  Lu et al. (IEEE Trans. Inform. Theory, 2005)
to achieve this rate.
A lower bound on the decoding delay of maximal-rate complex orthogonal designs 
has been obtained by Adams et al. (IEEE Trans. Inform. Theory, 2007) and it is observed that Liang's construction achieves the bound on delay for $n$ equal to $1$ and $3$ modulo $4$ while Lu et al.'s construction achieves the bound for $n=0,1,3$ mod $4$. For $n=2$ mod $4$, Adams et al. (IEEE Trans. Inform. Theory, 2010) have shown that the minimal decoding delay is twice the lower bound, in which case, both Liang's and Lu at al.'s construction achieve the minimum decoding delay. % when $n=2$ mod $4$.
For large value of $n$, it is observed that the rate is close to half and the decoding delay is very large. A class of rate-$\frac{1}{2}$ codes with low decoding delay for all $n$
has been constructed by Tarokh et al. (IEEE Trans. Inform. Theory, 1999).
% have constructed a class of rate-$\frac{1}{2}$ codes with low decoding delay for all $n$. 
In this paper, another class of rate-$\frac{1}{2}$ codes is constructed for all $n$ in which case the decoding delay is half the decoding delay of the rate-$\frac{1}{2}$ codes given by Tarokh et al. This is achieved by giving first a general construction of square real orthogonal designs which includes as special cases the well-known constructions of Adams, Lax and Phillips and the construction of Geramita and Pullman, and then making use of it to obtain the desired rate-$\frac{1}{2}$ codes. For the case of 9 transmit antennas, the proposed rate-$\frac{1}{2}$ code is shown to be of minimal-delay. The proposed construction results in designs with zero entries which may have high peak-to-average power ratio and it is shown that by appropriate post-multiplication, a design with no zero entry can be obtained with no change in the code parameters.
\end{abstract}
%%%%%%%%%%%%%%%%%%%%%%%

%%%%%%%%%%%%%%%%%%%%%%%%%%%%%%%%%%%%%
\section{Introduction}
\label{sec1}
Space-time block codes (STBCs) from complex orthogonal designs (CODs) have been widely studied for square designs, since they correspond to minimum-delay codes for co-located multiple-antenna coherent communication systems. However, non-square designs naturally appear 
in the following  situations.
\begin{enumerate}
\item In coherent co-located MIMO systems, for a specified number of transmit antennas, non-square designs can give much higher rate than the square designs~\cite{Lia}.
\item In non-coherent MIMO systems with non-differential detection, non-square designs with $p=2n$ lead to low decoding complexity STBCs~\cite{TaK}.
\item Space-time-frequency codes can be viewed as non-square designs~\cite{JSKR}.
\item In distributed space-time coding for relay channels, rectangular designs appear naturally~\cite{JiH}.
\end{enumerate}
\begin{defn}
A {\textit{complex orthogonal design}} (COD) in complex variables $x_0,x_1,\cdots,x_{k-1}$ is a $p\times n$ matrix $G$ % in $k$ complex variables $x_0,x_1,\cdots,x_{k-1}$ 
with entries $0,\pm x_0,\pm x_1,\cdots,\pm x_{k-1}$, their complex conjugates $\pm x_0^*,\pm x_1^*,\cdots,\pm x_{k-1}^*$ such that 
$G^\mathcal{H}G=({\vert x_0\vert}^2 +{\vert x_1\vert}^2+\cdots+{\vert x_{k-1}\vert}^2)I_n$, where $G^\mathcal{H}$ is the complex conjugate transpose of $G$ and $I_n$ is the $n\times n$ identity matrix.
The matrix $G$ is also said to be a $[p,n,k]$ COD.
When $x_0,\cdots,x_{k-1}$ are real variables, the corresponding design is called real orthogonal design (ROD). 
\end{defn}
An orthogonal design (OD) will always mean both real or complex orthogonal design.
The rate of a $[p,n,k]$ OD $G$ (defined as the number of complex symbols per channel use )
is $\frac{k}{p}$ and $p$ is called the decoding delay of the OD $G$.

The main problem in the construction of orthogonal designs is to construct a $p\times n$ orthogonal design (for given $n$) in $k$ variables which maximizes the rate $\frac{k}{p}$  and then to find a $p\times n$ orthogonal design with maximal rate which minimizes $p$.

It has been noted that the rate of the square ODs is very low for large number of antennas. 
Let $n$ be a positive integer and $\rho$ be a function (known as Hurwitz-Radon function) given by
the following formula: write $n=2^a(2b+1),a=4c+d;~a,b,c$ and $d$ are integers with $0\leq d\leq 3$, then
\begin{equation}
\label{rs}
\rho(n)=8c+2^d.
\end{equation}
% where $n=2^a(2b+1),a=4c+d;a,b,c$ and $d$ are integers with
% $0\leq d\leq 3$.
It is known that~\cite{TiH,TJC,ALP} the maximal rate of a square ROD for $n$ transmit antennas is  $\frac{\rho(n)}{n}$ while that of a square COD $\frac{a+1}{n}$.

As the square ODs are not bandwidth efficient, it is natural to study non-square orthogonal designs expecting that there may exist codes with high rate. It is known~\cite{TJC} that there always exists a rate-1 ROD for any number of transmit antennas. In fact, all rate-1 RODs can be obtained from square RODs of appropriate size.
The minimum decoding delay of a rate-1 ROD for $n$ transmit antennas~\cite{TJC} is $\nu(n)$
which is given by the following formula:

{\footnotesize
\begin{eqnarray}
\label{nun}
\begin{array}{l}
\nu(n)=2^{\delta(n)}\mbox{   where   }\\
 \delta(n)=
    \begin{cases}
4s      & \text{ if $n=8s+1$ } \\
4s+1    & \text{ if $n=8s+2$ } \\
4s+2    & \text{ if $n=8s+3$ or $8s+4$ }\\
4s+3    & \text{ if $n=8s+5,8s+6,8s+7$ or $8s+8$. }
  \end{cases}
\end{array}
\end{eqnarray}
}

%We often use this function in the rest of the paper
On the other hand, it is not known, in general, the maximal rate of a complex orthogonal design which admits as entries linear combination of several complex variables for arbitrary number of antennas. However, it is shown by Liang~\cite{Lia} that the maximal rate of a COD
is $\frac{t+1}{2t}$ whenever number of transmit antennas is $2t-1$ or $2t$.
Construction of maximal-rate CODs given by Liang ~\cite{Lia} is stated in the form of an algorithm while Lu et al~\cite{LFX} have constructed these codes by concatenating several matrices of smaller size.
%The former construction is algorithmic in nature while the latter one is based on {\it patch-up} of several matrices.
The following theorem describes the minimum decoding delay of the maximal-rate non-square CODs:
% is, in general, not known. The following theorem states what is known about the minimum delay of these code. For details, see \cite{AKP}.
\begin{table*}
\caption{The comparison of maximum rate achieving codes and rate 1/2  codes}
\label{recparameters}
\begin{center}
\begin{tabular}{c||rrrrrrrrrrrrrrrr}
$n$                     &5 &6 &7    &8  &9 &10  &11 &12 &13 &14 &15 &16 \\ \hline \hline
Decoding delay of $RH_n$ &8 &8 &8 &8 &16 &32 &64 &64 &128 &128 &128 &128 \\ \hline
Decoding delay of $TJC_n$ &16 &16 &16 &16 &32 &64 &128 &128 &256 &256 &256 &256 \\ \hline
\mbox{ Decoding delay of } $L_n$ &15 &30 &56 &56 &210 &420 &792 &792  &3003 & 6006& 11440& 11440\\ \hline \hline
$\mbox{ Rate of }RH_n$          &1/2 &1/2 &1/2 &1/2 &1/2 &1/2 &1/2 &1/2 &1/2 &1/2 & 1/2 &1/2\\\hline
$\mbox{ Rate of }TJC_n$          &1/2 &1/2 &1/2 &1/2 &1/2 &1/2 &1/2 &1/2 &1/2 &1/2 &1/2 &1/2\\\hline
$\mbox{ Rate of } L_n$          &2/3 &2/3 &5/8 &5/8 &3/5 &3/5 &7/12 &7/12 &4/7 &4/7 &9/16 &9/16 \\\hline\hline

\vspace*{0.1cm}
\end{tabular}
\end{center}
\hrule
\end{table*}
% The minimum decoding delay for maximal-rate CODs is, in general, not known.
% The following theorem states what is known about the minimum delay of these codes. 
\begin{theorem}[~\cite{AKP,AKM}] 
\label{delaycod}
A tight lower bound on the decoding delay of a maximum-rate COD for $n$ antennas is $\binom{2m}{m-1}$ for $n=2m-1$ or $n=2m$. Moreover, if $n$ is congruent to $0,1$ or $3$ modulo $4$, then this lower bound on decoding delay is achievable. If $n$ is congruent to $2$ modulo $4$, the minimum decoding delay is twice the lower bound.
\end{theorem}

As the rate of the maximal-rate codes is close to $\frac{1}{2}$ for large number of antennas
and the decoding delay of these codes is large, it is important to know whether there exists rate-$\frac{1}{2}$ codes with low decoding delay.
%The importance of the construction of rate-1/2 CODs with low decoding delay has also been recognized by Adams et al.\cite{AKP}. 
The importance of determining the delay of rate-$\frac{1}{2}$ CODs has also been noted by Adams et al~\cite{AKP}. % for the same reasons.

 A construction of rate-$\frac{1}{2}$ codes for any number of antennas is given by Tarokh et al.~\cite{TJC}. Their construction is simple: start with a rate-1 ROD $\mathcal{O}$ for $n$ antennas in $\nu(n)$ variables $x_0,x_2,\cdots,x_{\nu(n)-1}$,
and then form the following matrix
\begin{eqnarray}  
\label{TJC_half}
TJC_n=\frac{1}{\sqrt{2}}\left[\begin{array}{c}
\mathcal{O}\\
\mathcal{O}^*
    \end{array}\right]
\end{eqnarray}
where $\mathcal{O}^*$ is obtained from $\mathcal{O}$ by replacing each variable with its complex conjugate and $\nu(n)$ is given by \eqref{nun}.  Note that the number of rows in $TJC_n$ is $2\nu(n)$ and each variable appears twice along each column of the matrix.

We define a {\textit {$\lambda$-scaled complex orthogonal design}}, for a positive integer $\lambda$, ($\lambda$-scaled-COD) $G$ as a $p\times n$ orthogonal matrix
with non-zero entries the indeterminates $\pm x_0,\pm x_1,\cdots,\pm x_{k-1}$, their conjugates
or all the non-zero entries in a subset of columns of the matrix are of the form  $\pm\frac{1}{\sqrt{\lambda}}x_i,\pm\frac{1}{\sqrt{\lambda}}x_i^*,i=0,1,\cdots,k-1$.
Notice that a $\lambda$-scaled COD corresponds to a COD if $\lambda =1$.
In columns with scaling by $\frac{1}{\sqrt{\lambda}}$, all the variables appear exactly $\lambda$ times. 
In other words, lambda scaling (where Lambda ($\lambda$) is an integer greater than $1$) of a complex orthogonal design allows  all the non-zero entries in a subset of columns of the matrix
to take values from the set $\{\pm \frac{1}{\sqrt{\lambda}}x_i,\pm \frac{1}{\sqrt{\lambda}}x_i^*,i=0,1,\cdots,k-1\}$.
 It must be noted that scaling of a design is not something new as it has been already used by Seberry et al.~\cite{SSW} to construct orthogonal designs with fewer zeros. 
In this paper, $\lambda$ is always $2$ and call these codes simply {\textit{scaled-COD}}s. 

In the most general case, a {\textit{linear-processing complex orthogonal design}} (LPCOD) is a $p\times n$ orthogonal matrix $G$ in variables $x_0,x_1,\cdots,x_{k-1}$ such that each non-zero entry of the matrix is a complex linear combinations of the variables $x_0,x_1,\cdots,x_{k-1}$ and their conjugates.
If $x_0,x_1,\cdots,x_{k-1}$ are real variables, then the corresponding design is called {\textit{linear-processing real orthogonal design}} (LPROD).
Note that a scaled-COD is an LPCOD, but not conversely.
An example~\cite{TJC} of an LPCOD which is not a scaled-COD is the following code:
\begin{eqnarray*}
\left[
\begin{array}{cccc}
x_0   & x_1   &\frac{x_2}{\sqrt{2}} & \frac{x_2}{\sqrt{2}}     \\
-x_1^* & x_0^* &\frac{x_2}{\sqrt{2}} & \frac{-x_2}{\sqrt{2}}   \\
\frac{x_2^*}{\sqrt{2}} & \frac{x_2^*}{\sqrt{2}}  &\frac{(-x_0-x_0^*+x_1-x_1^*)}{2}& \frac{(x_0-x_0^*-x_1-x_1^*)}{2} \\ 
\frac{x_2^*}{\sqrt{2}} & \frac{-x_2^*}{\sqrt{2}}  &\frac{(x_0-x_0^*+x_1+x_1^*)}{2} &-\frac{(x_0+x_0^*+x_1-x_1^*)}{2}
\end{array}
\right].
\end{eqnarray*}
It has been observed that the decoding delay of the rate-$\frac{1}{2}$ codes obtained by the construction~\eqref{TJC_half} is not the best possible: for example, the following code for $8$ antennas
\begin{eqnarray*}
\left[\begin{array}{r@{\hspace{0.9pt}}r@{\hspace{0.9pt}}r@{\hspace{0.9pt}}r@{\hspace{0.9pt}}r@{\hspace{0.9pt}}r@{\hspace{0.9pt}}r@{\hspace{0.9pt}}r@{\hspace{0.9pt}}}
    x_0   &-x_1^*  &-x_2^*& 0    &-x_3^* & 0     & 0    & 0 \\
    x_1   & x_0^*  & 0    &-x_2^*& 0     &-x_3^* & 0    & 0 \\
    x_2   & 0      & x_0^*& x_1^*& 0     & 0     &-x_3^*& 0 \\
     0    & x_2    &-x_1  & x_0  & 0     & 0     & 0    &-x_3^* \\
    x_3   & 0      & 0    & 0    & x_0^* & x_1^* & x_2^*& 0 \\
     0    & x_3    & 0    & 0    &-x_1   & x_0   & 0    & x_2^*\\
     0    & 0      & x_3  & 0    &-x_2   & 0     & x_0  &-x_1^*\\
     0    & 0      & 0    & x_3  & 0     &-x_2   & x_1  & x_0^*\\
\end{array}\right]
%}
\end{eqnarray*}
is a rate-$\frac{1}{2}$ COD with decoding delay $8$, whereas the corresponding rate-$\frac{1}{2}$ code given by the construction~\eqref{TJC_half}
has decoding delay $16$. 
This indicates that there may exist rate-$\frac{1}{2}$ scaled-COD for any number of antennas with half the decoding delay of the rate-$\frac{1}{2}$ code given by~\eqref{TJC_half}.

In this paper, we provide an explicit construction of rate-$\frac{1}{2}$ scaled-COD for any number of transmit antennas, say $n$, with decoding delay $\nu(n)$. % which is defined in \eqref{nun}.
Table~\ref{recparameters} gives a comparison of the three classes of codes, namely, maximal rate CODs (denoted by $L_n$), rate-$\frac{1}{2}$ scaled-CODs ($TJC_n$) and the rate-$\frac{1}{2}$ codes of this paper (denoted by $RH_n$). It shows that for large values of $n,$ but for a marginal decrease in the rate with respect to $L_n,$ the codes of this paper are the best codes known to date with respect to decoding delay.

As a byproduct of the above mentioned  construction, a general construction of square RODs is presented which includes as special cases the well-known constructions of Adams, Lax and Phillips~\cite{ALP} and the construction of Geramita and Pullman~\cite{GeP}.

Though the minimum value of the decoding delay of the maximal-rate CODs is well-known~\cite{AKP}, nothing is known about the minimal-delay of the rate-$\frac{1}{2}$ scaled-CODs. However, we have only been able to show that the decoding delay of the proposed rate-$\frac{1}{2}$ code for $9$ transmit antennas is minimum.

Zero entries in a design increase the peak-to-average power ratio (PAPR) in the transmitted signal and it is preferred not to have any zero entry in the design. This problem has been addressed for square and non-square orthogonal designs~\cite{DaR,TWSMS,SSW}. Our initial construction of rate-$\frac{1}{2}$ scaled-CODs contain zero entries in the design matrix which will lead to higher PAPR in contrast to the designs $TJC_n$ given by \eqref{TJC_half}. However, we show that by post-multiplication of appropriate matrices, our construction leads to designs with no zero entry without any change in the parameters of the designs.

The remaining part of the paper is organized as follows: In Section \ref{sec2}, we present the main  result of the paper given by Theorem \ref{rate12cod}.  
For the special case of 9 transmit antennas, in Section \ref{sec3}, it is shown that our construction is of  minimal delay. In Section \ref{sec4}, we show that the codes discussed so far can be made to have no zero entry in it by appropriate preprocessing without affecting the parameters of the design. Concluding remarks constitute Section \ref{sec5}. 

\section{A Construction of rate-$\frac{1}{2}$ Scaled Complex Orthogonal Designs}
\label{sec2}
Construction of the rate-$\frac{1}{2}$ codes is obtained in the following three steps:\\
STEP 1:  Construction of a new set of square RODs (Subsection \ref{subsec2-1}).\\
STEP 2:  Construction of two new sets of rate-1 RODs from the square RODs of STEP 1 (Subsection \ref{subsec2-2}). \\
STEP 3:  Construction of low-delay rate-$\frac{1}{2}$ scaled-CODs using rate-1 RODs (Subsection \ref{subsec2-3}).

Before explaining these steps, we first build up some preliminary results needed to describe these steps. 
\subsection{Mathematical Preliminaries}
$\mathbb{F}_2$ denotes the finite field consisting of two elements with two binary operations addition and multiplication denoted by $b_1\oplus b_2$ and $b_1b_2$ respectively, $b_1,b_2\in \mathbb{F}_2$.
Let $b_1+b_2$ and $\bar{b}_1$ represent respectively the logical disjunction (OR) of $b_1$ and $b_2$ and complement or negation of $b_1$.\\
%  Note that
% \begin{eqnarray}
% \begin{array}{c}
% \label{binaryop}
% b_1+b_2=b_1\oplus b_2 \oplus b_1b_2,~~~~
% \bar{b}_1=1\oplus b_1.
%  \end{array}
% \end{eqnarray}
Let $l$ be a non-zero positive integer and $Z_l=\{0,1,\cdots,l-1\}$. We identify $Z_{2^a}$ with the set $\f_2^a$ of $a$-tuple binary vectors in the standard way, i.e., 
any element of $Z_{2^a}$ is identified with its radix-2 representation vectors (of length $a$)
via the correspondence:  % using the following correspondence:
$x\in Z_{2^a} \leftrightarrow (x_{a-1},\cdots,x_0)\in\mathbb{F}_2^a$ such that $x=\sum_{j=0}^{a-1} x_j2^j,x_j\in \mathbb{F}_2$.
For convenience, depending on the context, the  set $Z_{2^a}$ is used as the set of positive integers and sometimes as the set of binary vectors.\\
For $x=(x_{a-1},\cdots,x_0),y=(y_{a-1},\cdots,y_0),x_i,y_i\in \f_2,i=0,1,\cdots,a-1$, the component-wise modulo-2 addition and the component-wise multiplication  of $x$ and $y$ are denoted by $x\oplus y$ and $x\cdot y$ respectively.
We have $x\oplus y=(x_{a-1}\oplus y_{a-1},\cdots,x_0\oplus y_0),x\cdot y=(x_{a-1} y_{a-1},\cdots,x_0y_0)$.
% \begin{align*}
% x\oplus y&=(x_{a-1}\oplus y_{a-1},\cdots,x_0\oplus y_0),\\
% x\cdot y&=(x_{a-1} y_{a-1},\cdots,x_0y_0).
% \end{align*}
The two's complement of a number $x\in Z_{2^a}$, denoted by $\overline{x}$ is defined as the value obtained by subtracting the number from a large power of two (specifically, from $2^a$ for an $a$-bit two's complement) i.e., $\overline{x}=2^a-x$.\\
The Hamming weight of $x$, denoted by $\vert x\vert$ is the number of $1$ in the binary representation of $x$.
For two integers $i,j$, we use the notation $i\equiv j$, to mean $i-j=0 \mod 2$.

For any matrix of size $n_1\times n_2,$ the rows and the columns of the matrix are labeled by the elements of $\{0,1,\cdots,n_1-1\}$ and $\{0,1,\cdots,n_2-1\}$ respectively. If $M$ is a $p\times n$ matrix in $k$ real variables $x_0,x_1, x_2,\cdots, x_{k-1}$, such that each non-zero entry of the matrix is $x_i$ or $-x_i$ for some $i\in\{0,1,\cdots,k-1\}$, it is not necessary that $M$ is an ROD. For example,
$\begin{bmatrix}
    x_0   &x_1    \\
    x_1   & x_0  
\end{bmatrix}
$ is not an ROD. 
A sub-matrix $M_2$ of size $2\times 2,$  constructed by choosing any two rows and any two columns of $M$ is called {\textit{proper}} if
\begin{itemize}
\item none of the entries of $M_2$ is zero and
\item it contains exactly two distinct variables.
\end{itemize}
\begin{example}
Consider the following matrix in three real variables $x_0,x_1$ and $x_2$
{\small
\begin{equation}
\label{m4}
 \begin{bmatrix}
    x_0   &-x_1   &-x_2  & 0     \\
    x_1   & x_0   & 0    &-x_2     \\
    x_2   & 0     & x_0  & x_1     \\
     0    & x_2   &-x_1  & x_0
\end{bmatrix}.
\end{equation}}
The sub-matrix 
$\begin{bmatrix}
    x_1   &-x_2    \\
    x_2   & x_1   
  \end{bmatrix}
$ is {\textit{proper}} while
$
\begin{bmatrix}
    x_3   & 0 \\
    0     & x_3 
  \end{bmatrix}
$
is not.
\end{example}
If $M(i,j)\neq 0$, then we write $\vert M(i,j)\vert= k$ whenever $M(i,j)=x_k$ or $-x_k$. 

It is easy to see that the following two statements are equivalent:\\
 1) $M$ is an ROD.\\
 2) (i) Each variable appears exactly once along each column\\
\hspace*{.90cm}of $M$ and at most once along each row of $M$,\\
\hspace*{.40cm}(ii) if for some $i,j,j^\prime$, $M(i,j)\neq 0$ and $M(i,j^\prime)\neq 0$,\\
 \hspace*{.90cm}then there exists $i^\prime$ such that $\vert M(i,j)\vert=\vert M(i^\prime,j^\prime)\vert$ \\
 \hspace*{.80cm} and $\vert M(i,j^\prime)\vert=\vert M(i^\prime,j)\vert$,\\
 \hspace*{.40cm}(iii) any proper $2\times 2$ sub-matrix of $M$ is an ROD.
\subsection{STEP 1: Construction of a new class of square RODs}
\label{subsec2-1}
Square RODs have been constructed by several authors, for example, Adams et al.~\cite{ALP} and Geramita et al~\cite{GeP}. All these designs are constructed recursively and the basic building blocks of these designs are the RODs of order $1,2,4$ and $8$.
In this subsection, we take a different approach towards the construction of square RODs and it leads to a new class of RODs of which the constructions in~\cite{ALP} and~\cite{GeP} are special cases.
For any ROD, a non-zero entry of it is characterized by a pair of two integers, the first component of which takes value from the set $\{+1,-1\}$ denoting the sign of the entry
while the second component represents the variable at that entry. For example, the $(0,0)$-th entry of \eqref{m4} corresponds to the pair $(1,0)$ while the $(0,1)$-th entry corresponds to $(-1,1)$.

For a square ROD $B_t$ of order $t$ in $k$ real variables $x_0,\cdots,x_{k-1}$, we define two functions $\mu_t$ and $\lambda_t$ on the set $Z_t\times Z_t$ with ${\mu}_t(i,j)\in\{1,-1\}$ and $\lambda_t(i,j)\in Z_k,i,j\in Z_t$ such that $B_t(i,j)={\mu}_t(i,j)x_{\lambda_t(i,j)}$ whenever $B_t(i,j)\neq 0$.
It is straightforward to see that $B_t$ is uniquely determined by $\mu_t$ and $\lambda_t$.
However, any arbitrary choice of these two functions will not lead to a square ROD. Therefore
the approach we take is identifying a pair of functions $\mu_t$ and $\lambda_t$ that results in a square ROD. 
Let 
\begin{equation}
\label{gammamap}
\gamma_t: Z_{\rho(t)}\rightarrow Z_t 
\end{equation}
be an injective map defined on $Z_{\rho(t)}$ with the image denoted by
$\hat{Z}_{\rho(t)}=\gamma_t(Z_{\rho(t)})$ 
and  
\begin{equation}
\label{psimap}
\psi_t: \hat{Z}_{\rho(t)}\rightarrow Z_t 
\end{equation}
be another injective map defined on $\hat{Z}_{\rho(t)}.$ $\rho(t)$ is given by \eqref{rs}.

In the following theorem, we define two maps $\mu_t$ and $\lambda_t$ 
in terms of the maps \eqref{gammamap} and \eqref{psimap} and identify the conditions 
so that the resulting $B_t$ becomes a square ROD. 
%for the resulting $B_t$ to be a square ROD. 
%%%%%%%%%%%%%%%%%%%

\begin{theorem}
\label{ratesquarerod}
Let $t=2^a$. Construct a square matrix $B_t$ of order $t$ in $\rho(t)$ variables $x_0,\cdots,x_{\rho(t)-1}$ as follows:
\begin{eqnarray*}
 B_t(i,j)=
    \begin{cases}
{\mu}_t(i,j)x_{\lambda_t(i,j)}      & \text{ if $i\oplus j\in \hat{Z}_{\rho(t)}$ } \\
 0 & \text{ otherwise, }
\end{cases}
\end{eqnarray*}
where $\mu_t(i,j)=(-1)^{\vert i\cdot\psi_t(i\oplus j)\vert}$ and 
$\lambda_t(i,j)=\gamma_t^{-1}(i\oplus j)$. 
Suppose $\forall x,y\in\hat{Z}_{\rho(t)}, x\neq y$,
\begin{equation}
\label{oddcondition}
\vert (\psi_t(x)\oplus \psi_t(y))\cdot(x\oplus y)\vert
\mbox{ is odd.} 
\end{equation}
Then $B_t$ is a square ROD of size $[t,t,\rho(t)]$.

\end{theorem}
%%%%%%%%%%
\begin{proof}
By definition, 
each of the variables $x_0,x_1,\cdots,x_{\rho(t)-1}$ appears exactly once in each column of the matrix and at most once along each row of $B_t$.
Secondly, assume that $B_t(i,j)\neq 0$ and $B_t(i,j^\prime)\neq 0$, then we show that there exists $i^\prime$ such that
\begin{eqnarray*}
\vert B_t(i,j)\vert=\vert B_t(i^\prime,j^\prime)\vert \mbox{ and  }
\vert B_t(i,j^\prime)\vert = \vert B_t(i^\prime,j)\vert.
\end{eqnarray*}

Let $i^\prime=i\oplus j\oplus j^\prime$. Then 
$\vert B_t(i,j)\vert=\gamma_t^{-1}(i\oplus j)$ and
$\vert B_t(i^\prime,j^\prime)\vert=\gamma_t^{-1}(i^\prime\oplus j^\prime)=\gamma_t^{-1}(i\oplus j)$,
therefore $\vert B_t(i,j)\vert=\vert B_t(i^\prime,j^\prime)\vert$.
Similarly, $\vert B_t(i,j^\prime)\vert=\vert B_t(i^\prime,j)\vert$.

Thirdly, we show that any proper $2\times 2$ sub-matrix of $B_t$ is an ROD, that is, $\mu_t(i,j)\cdot \mu_t(i,j^\prime)\cdot \mu_t(i^\prime,j)\cdot \mu_t(i^\prime,j^\prime)=-1$
whenever $i+i^\prime=j\oplus j^\prime$. 
Now
\begin{align*}
\label{inter}
%&\mu_t(i,j)\cdot \mu_t(i,j^\prime)\cdot \mu_t(i^\prime,j)\cdot \mu_t(i^\prime,j^\prime)\\
&\vert i\cdot\psi_t(i\oplus j) \vert +\vert i\cdot\psi_t(i\oplus j^\prime) \vert
+\vert i^\prime\cdot\psi_t(i^\prime\oplus j) \vert\\
&~~~~~+\vert i^\prime\cdot\psi_t(i^\prime\oplus j^\prime)\vert\\
&\equiv \vert (i\oplus i^\prime)\cdot(\psi_t(i\oplus j) \oplus \psi_t(i^\prime\oplus j)) \vert\\
&\equiv \vert ((i\oplus j)\oplus(i^\prime\oplus j))\cdot(\psi_t(i\oplus j) \oplus \psi_t(i^\prime\oplus j)) \vert
\end{align*}
is an odd number.
Therefore, $\mu_t(i,j)\cdot \mu_t(i,j^\prime)\cdot \mu_t(i^\prime,j)\cdot \mu_t(i^\prime,j^\prime)=-1$.

% We can write the above expression as $\vert (i\oplus i^\prime)\cdot(\psi_t(i\oplus j) \oplus \psi_t(i^\prime\oplus j)) \vert$.
% But $i\oplus i^\prime=(i\oplus j)\oplus(i^\prime\oplus j)$ and both $i\oplus j$ and $i^\prime\oplus j$ are the elements of $\hat{Z}_{\rho(t)}$. Therefore, the value of the expression given by \eqref{inter} is an odd number.
\end{proof}

We now construct the maps $\psi_t$ and $\gamma_t$ explicitly such that \eqref{oddcondition} is satisfied.
The map $\gamma_t:Z_{\rho(t)}\rightarrow Z_t$ is given by
{\small{
\begin{eqnarray}
\label{gammat}
 \gamma_t(i)=
    \begin{cases}
i   & \text{ if  $0\leq i\leq 7$ } \\
2^{4l-1}\cdot\hat{\gamma}(m)&\text{ if $i\geq 8,i=8l+m$, $0\leq m\leq 7$ }
  \end{cases}
  \end{eqnarray}
}}
\begin{equation*}  
\mbox{ where }\hat{\gamma}=\left(\begin{array}{cccccccc}
 0& 1 &2 &3&4&5&6&7 \\
 1& 2 &4 &7&8&11&13&14
\end{array}\right),
\end{equation*}
that is, $\hat{\gamma}(0)=1,\cdots,\hat{\gamma}(7)=14$.\\
Let $F=\hat{\gamma}(Z_8)$.
For an element $x\in \hat{Z}_{\rho(t)}$, either $x\in Z_8$ or $x=2^{4y-1}z$ for some $y\in\mathbb{N}\setminus\{0\}$ and $z\in F$. Note that $\hat{Z}_{\rho(t)}=\gamma_t(Z_{\rho(t)})$.

We now define a map $\phi:\hat{Z}_{\rho(t)}\rightarrow Z_t$ given by
\begin{eqnarray}
\phi(x)&=&
    \begin{cases}
\phi_1(x)  & \text{ if  $x\in Z_8$} \\
2^{4y-1}\cdot\phi_2(z)&\text{ if $x=2^{4y-1}z,~z\in F$}
  \end{cases}
  \end{eqnarray}
where $\phi_1: Z_8\rightarrow Z_8$ be the map given by
 \begin{equation}  
\label{phi_07}
\phi_1=\left(\begin{array}{cccccccccccc}
    0   & 1 &2 &3&4&5&6&7   \\
    0   & 1 &2 &3&4&7&5&6
    \end{array}\right)
\end{equation}
and $\phi_2: F\rightarrow Z_{16}$ be an injective map given by
\begin{equation}  
\phi_2=\left(\begin{array}{cccccccccccc}
     1 &2 &4&7&8&11&13&14   \\
     1 &2 &4&6&8&15&10&12
    \end{array}\right).
\end{equation}
Let 
\begin{equation}
\label{mappsi}
\psi_t(x)=\overline{\phi(x)} \mbox{ in } \mathbb{F}_2^a~ \forall x\in \hat{Z}_{\rho(t)}.
\end{equation}
Note that $\overline{z}$ is two's complement of $z$.\\
In order to show that the map $\psi_t$ so constructed satisfies the condition of \eqref{oddcondition}, we need the following two results related to the maps $\phi_1$ and $\phi_2$.

\begin{figure*}
\small{
\begin{eqnarray}
\label{R16}
R_{16}=
\left[\begin{array} % {rrrrrrrrrrrrrrrr}
{ r @{\hspace{.2pt}} r @{\hspace{.2pt}} r @{\hspace{.2pt}} r @{\hspace{.2pt}}r @{\hspace{.2pt}} r @{\hspace{.2pt}} r @{\hspace{.2pt}} r @{\hspace{.2pt}}
r @{\hspace{.2pt}} r @{\hspace{.2pt}} r @{\hspace{.2pt}} r @{\hspace{.2pt}}r @{\hspace{.2pt}} r @{\hspace{.2pt}} r @{\hspace{.2pt}} r @{\hspace{.2pt}}}
  x_0&x_1&  x_2& x_3& x_4& x_5&  x_6 & x_7 & x_8 &  0 &  0 &  0 &  0 &  0 &  0 &  0\\
 -x_1 & x_0& -x_3 & x_2& -x_5 & x_4 & x_7 & -x_6 &  0 & x_8 &  0 &  0 &  0 &  0 &  0 &  0 \\
 -x_2 & x_3 & x_0 & -x_1 & -x_6 & -x_7 & x_4 & x_5 &  0 &  0 & x_8 &  0 &  0 &  0 &  0 &  0 \\
 -x_3 & -x_2 & x_1 & x_0 & -x_7 & x_6 & -x_5 & x_4 &  0 &  0 &  0 & x_8 &  0 &  0 &  0 &  0 \\
 -x_4 & x_5 & x_6 & x_7 & x_0 & -x_1 & -x_2 & -x_3 &  0 &  0 &  0 &  0 & x_8 &  0 &  0 &  0 \\
 -x_5 & -x_4 & x_7 & -x_6 & x_1 & x_0 & x_3 & -x_2 &  0 &  0 &  0 &  0 &  0 & x_8 &  0 &  0 \\
 -x_6 & -x_7 & -x_4 & x_5 & x_2 & -x_3 & x_0 & x_1 &  0 &  0 &  0 &  0 &  0 &  0 & x_8 &  0 \\
 -x_7 & x_6 & -x_5 & -x_4 & x_3 & x_2 & -x_1 & x_0 &  0 &  0 &  0 &  0 &  0 &  0 &  0 & x_8 \\
 -x_8 & 0 &  0 &  0 &  0 &  0 &  0 &  0 & x_0 & -x_1 & -x_2 & -x_3 & -x_4 & -x_5 & -x_6 & -x_7 \\
   0 & -x_8 &  0 &  0 &  0 &  0 &  0 &  0 & x_1 & x_0 & x_3 & -x_2 & x_5 & -x_4 & -x_7 & x_6 \\
   0 &  0 & -x_8 &  0 &  0 &  0 &  0 &  0 & x_2 & -x_3 & x_0 & x_1 & x_6 & x_7 & -x_4 & -x_5 \\
   0 &  0 &  0 & -x_8 &  0 &  0 &  0 &  0 & x_3 & x_2 & -x_1 & x_0 & x_7 & -x_6 & x_5 & -x_4 \\
   0 &  0 &  0 &  0 & -x_8 &  0 &  0 &  0 & x_4 & -x_5 & -x_6 & -x_7 & x_0 & x_1 & x_2 & x_3 \\
   0 &  0 &  0 &  0 &  0 & -x_8 &  0 &  0 & x_5 & x_4 & -x_7 & x_6 & -x_1 & x_0 & -x_3 & x_2 \\
   0 &  0 &  0 &  0 &  0 &  0 & -x_8 &  0 & x_6 & x_7 & x_4 & -x_5 & -x_2 & x_3 & x_0 & -x_1 \\
   0 &  0 &  0 &  0 &  0 &  0 &  0 & -x_8 & x_7 & -x_6 & x_5 & x_4 & -x_3 & -x_2 & x_1 & x_0 
\end{array}\right]
\end{eqnarray}
}
 \scriptsize{
\begin{eqnarray}
\label{R32}
%R_{32}=
\left[\begin{array}{ r @{\hspace{.2pt}} r @{\hspace{.2pt}} r @{\hspace{.2pt}} r @{\hspace{.2pt}}r @{\hspace{.2pt}} r @{\hspace{.2pt}} r @{\hspace{.2pt}} r @{\hspace{.2pt}}
r @{\hspace{.2pt}} r @{\hspace{.2pt}} r @{\hspace{.2pt}} r @{\hspace{.2pt}}r @{\hspace{.2pt}} r @{\hspace{.2pt}} r @{\hspace{.2pt}} r @{\hspace{.2pt}} r @{\hspace{.2pt}} r @{\hspace{.2pt}} r @{\hspace{.2pt}} r @{\hspace{.2pt}}r @{\hspace{.2pt}} r @{\hspace{.2pt}} r @{\hspace{.2pt}} r @{\hspace{.2pt}}r @{\hspace{.2pt}} r @{\hspace{.2pt}} r @{\hspace{.2pt}} r @{\hspace{.2pt}}r @{\hspace{.2pt}} r @{\hspace{.2pt}} r @{\hspace{.2pt}} r @{\hspace{.2pt}}}
  x_0 & x_1 & x_2 & x_3 & x_4 & x_5 & x_6 & x_7 & x_8 &  0 &  0 &  0 &  0 &  0 &  0 &  0 & x_9 &  0 &  0 &  0 &  0 &  0 &  0 &  0 &  0 &  0 &  0 &  0 &  0 &  0 &  0 &  0 \\
 -x_1 & x_0 & -x_3 & x_2 & -x_5 & x_4 & x_7 & -x_6 &  0 & x_8 &  0 &  0 &  0 &  0 &  0 &  0 & 0 & x_9 &  0 &  0 &  0 &  0 &  0 &  0 &  0 &  0 &  0 &  0 &  0 &  0 &  0 &  0 \\
 -x_2 & x_3 & x_0 & -x_1 & -x_6 & -x_7 & x_4 & x_5 &  0 &  0 & x_8 &  0 &  0 &  0 &  0 &  0 & 0 &  0 & x_9 &  0 &  0 &  0 &  0 &  0 &  0 &  0 &  0 &  0 &  0 &  0 &  0 &  0 \\
 -x_3 & -x_2 & x_1 & x_0 & -x_7 & x_6 & -x_5 & x_4 &  0 &  0 &  0 & x_8 &  0 &  0 &  0 &  0 & 0 &  0 &  0 & x_9 &  0 &  0 &  0 &  0 &  0 &  0 &  0 &  0 &  0 &  0 &  0 &  0 \\
 -x_4 & x_5 & x_6 & x_7 & x_0 & -x_1 & -x_2 & -x_3 &  0 &  0 &  0 &  0 & x_8 &  0 &  0 &  0 & 0 &  0 &  0 &  0 & x_9 &  0 &  0 &  0 &  0 &  0 &  0 &  0 &  0 &  0 &  0 &  0 \\
 -x_5 & -x_4 & x_7 & -x_6 & x_1 & x_0 & x_3 & -x_2 &  0 &  0 &  0 &  0 &  0 & x_8 &  0 &  0 & 0 &  0 &  0 &  0 &  0 & x_9 &  0 &  0 &  0 &  0 &  0 &  0 &  0 &  0 &  0 &  0 \\
 -x_6 & -x_7 & -x_4 & x_5 & x_2 & -x_3 & x_0 & x_1 &  0 &  0 &  0 &  0 &  0 &  0 & x_8 &  0 & 0 &  0 &  0 &  0 &  0 &  0 & x_9 &  0 &  0 &  0 &  0 &  0 &  0 &  0 &  0 &  0 \\
 -x_7 & x_6 & -x_5 & -x_4 & x_3 & x_2 & -x_1 & x_0 &  0 &  0 &  0 &  0 &  0 &  0 &  0 & x_8 & 0 &  0 &  0 &  0 &  0 &  0 &  0 & x_9 &  0 &  0 &  0 &  0 &  0 &  0 &  0 &  0 \\
 -x_8 &  0 &  0 &  0 &  0 &  0 &  0 &  0 & x_0 & -x_1 & -x_2 & -x_3 & -x_4 & -x_5 & -x_6 & -x_7 &  0 &  0 &  0 &  0 &  0 &  0 &  0 &  0 & x_9 &  0 &  0 &  0 &  0 &  0 &  0 &  0 \\
   0 & -x_8 &  0 &  0 &  0 &  0 &  0 &  0 & x_1 & x_0 & x_3 & -x_2 & x_5 & -x_4 & -x_7 & x_6 &  0 &  0 &  0 &  0 &  0 &  0 &  0 &  0 &  0 & x_9 &  0 &  0 &  0 &  0 &  0 &  0 \\
   0 &  0 & -x_8 &  0 &  0 &  0 &  0 &  0 & x_2 & -x_3 & x_0 & x_1 & x_6 & x_7 & -x_4 & -x_5 &  0 &  0 &  0 &  0 &  0 &  0 &  0 &  0 &  0 &  0 & x_9 &  0 &  0 &  0 &  0 &  0 \\
   0 &  0 &  0 & -x_8 &  0 &  0 &  0 &  0 & x_3 & x_2 & -x_1 & x_0 & x_7 & -x_6 & x_5 & -x_4 &  0 &  0 &  0 &  0 &  0 &  0 &  0 &  0 &  0 &  0 &  0 & x_9 &  0 &  0 &  0 &  0 \\
   0 &  0 &  0 &  0 & -x_8 &  0 &  0 &  0 & x_4 & -x_5 & -x_6 & -x_7 & x_0 & x_1 & x_2 & x_3 &  0 &  0 &  0 &  0 &  0 &  0 &  0 &  0 &  0 &  0 &  0 &  0 & x_9 &  0 &  0 &  0 \\
   0 &  0 &  0 &  0 &  0 & -x_8 &  0 &  0 & x_5 & x_4 & -x_7 & x_6 & -x_1 & x_0 & -x_3 & x_2 &  0 &  0 &  0 &  0 &  0 &  0 &  0 &  0 &  0 &  0 &  0 &  0 &  0 & x_9 &  0 &  0 \\
   0 &  0 &  0 &  0 &  0 &  0 & -x_8 &  0 & x_6 & x_7 & x_4 & -x_5 & -x_2 & x_3 & x_0 & -x_1 &  0 &  0 &  0 &  0 &  0 &  0 &  0 &  0 &  0 &  0 &  0 &  0 &  0 &  0 & x_9 &  0 \\
   0 &  0 &  0 &  0 &  0 &  0 &  0 & -x_8 & x_7 & -x_6 & x_5 & x_4 & -x_3 & -x_2 & x_1 & x_0 &  0 &  0 &  0 &  0 &  0 &  0 &  0 &  0 &  0 &  0 &  0 &  0 &  0 &  0 &  0 & x_9 \\
 -x_9 &  0 &  0 &  0 &  0 &  0 &  0 &  0 &  0 &  0 &  0 &  0 &  0 &  0 &  0 &  0 & x_0 & -x_1 & -x_2 & -x_3 & -x_4 & -x_5 & -x_6 & -x_7 & -x_8 &  0 &  0 &  0 &  0 &  0 &  0 &  0 \\
   0 & -x_9 &  0 &  0 &  0 &  0 &  0 &  0 &  0 &  0 &  0 &  0 &  0 &  0 &  0 &  0 & x_1 & x_0 & x_3 & -x_2 & x_5 & -x_4 & -x_7 & x_6 &  0 & -x_8 &  0 &  0 &  0 &  0 &  0 &  0 \\
   0 &  0 & -x_9 &  0 &  0 &  0 &  0 &  0 &  0 &  0 &  0 &  0 &  0 &  0 &  0 &  0 & x_2 & -x_3 & x_0 & x_1 & x_6 & x_7 & -x_4 & -x_5 &  0 &  0 & -x_8 &  0 &  0 &  0 &  0 &  0 \\
   0 &  0 &  0 & -x_9 &  0 &  0 &  0 &  0 &  0 &  0 &  0 &  0 &  0 &  0 &  0 &  0 & x_3 & x_2 & -x_1 & x_0 & x_7 & -x_6 & x_5 & -x_4 &  0 &  0 &  0 & -x_8 &  0 &  0 &  0 &  0 \\
   0 &  0 &  0 &  0 & -x_9 &  0 &  0 &  0 &  0 &  0 &  0 &  0 &  0 &  0 &  0 &  0 & x_4 & -x_5 & -x_6 & -x_7 & x_0 & x_1 & x_2 & x_3 &  0 &  0 &  0 &  0 & -x_8 &  0 &  0 &  0 \\
   0 &  0 &  0 &  0 &  0 & -x_9 &  0 &  0 &  0 &  0 &  0 &  0 &  0 &  0 &  0 &  0 & x_5 & x_4 & -x_7 & x_6 & -x_1 & x_0 & -x_3 & x_2 &  0 &  0 &  0 &  0 &  0 & -x_8 &  0 &  0 \\
   0 &  0 &  0 &  0 &  0 &  0 & -x_9 &  0 &  0 &  0 &  0 &  0 &  0 &  0 &  0 &  0 & x_6 & x_7 & x_4 & -x_5 & -x_2 & x_3 & x_0 & -x_1 &  0 &  0 &  0 &  0 &  0 &  0 & -x_8 &  0 \\
   0 &  0 &  0 &  0 &  0 &  0 &  0 & -x_9 &  0 &  0 &  0 &  0 &  0 &  0 &  0 &  0 & x_7 & -x_6 & x_5 & x_4 & -x_3 & -x_2 & x_1 & x_0 &  0 &  0 &  0 &  0 &  0 &  0 &  0 & -x_8 \\
   0 &  0 &  0 &  0 &  0 &  0 &  0 &  0 & -x_9 &  0 &  0 &  0 &  0 &  0 &  0 &  0 & x_8 &  0 &  0 &  0 &  0 &  0 &  0 &  0 & x_0 & x_1 & x_2 & x_3 & x_4 & x_5 & x_6 & x_7 \\
   0 &  0 &  0 &  0 &  0 &  0 &  0 &  0 &  0 & -x_9 &  0 &  0 &  0 &  0 &  0 &  0 &  0 & x_8 &  0 &  0 &  0 &  0 &  0 &  0 & -x_1 & x_0 & -x_3 & x_2 & -x_5 & x_4 & x_7 & -x_6 \\
   0 &  0 &  0 &  0 &  0 &  0 &  0 &  0 &  0 &  0 & -x_9 &  0 &  0 &  0 &  0 &  0 &  0 &  0 & x_8 &  0 &  0 &  0 &  0 &  0 & -x_2 & x_3 & x_0 & -x_1 & -x_6 & -x_7 & x_4 & x_5 \\
   0 &  0 &  0 &  0 &  0 &  0 &  0 &  0 &  0 &  0 &  0 & -x_9 &  0 &  0 &  0 &  0 &  0 &  0 & 0 & x_8 &  0 &  0 &  0 &  0 & -x_3 & -x_2 & x_1 & x_0 & -x_7 & x_6 & -x_5 & x_4 \\
   0 &  0 &  0 &  0 &  0 &  0 &  0 &  0 &  0 &  0 &  0 &  0 & -x_9 &  0 &  0 &  0 &  0 &  0 & 0 &  0 & x_8 &  0 &  0 &  0 & -x_4 & x_5 & x_6 & x_7 & x_0 & -x_1 & -x_2 & -x_3 \\
   0 &  0 &  0 &  0 &  0 &  0 &  0 &  0 &  0 &  0 &  0 &  0 &  0 & -x_9 &  0 &  0 &  0 &  0 & 0 &  0 &  0 & x_8 &  0 &  0 & -x_5 & -x_4 & x_7 & -x_6 & x_1 & x_0 & x_3 & -x_2 \\
   0 &  0 &  0 &  0 &  0 &  0 &  0 &  0 &  0 &  0 &  0 &  0 &  0 &  0 & -x_9 &  0 &  0 &  0 & 0 &  0 &  0 &  0 & x_8 &  0 & -x_6 & -x_7 & -x_4 & x_5 & x_2 & -x_3 & x_0 & x_1 \\
   0 &  0 &  0 &  0 &  0 &  0 &  0 &  0 &  0 &  0 &  0 &  0 &  0 &  0 &  0 & -x_9 &  0 &  0 & 0 &  0 &  0 &  0 &  0 & x_8 & -x_7 & x_6 & -x_5 & -x_4 & x_3 & x_2 & -x_1 & x_0
\end{array}\right]
\end{eqnarray}
}
\hrule
\end{figure*}

\begin{lemma}
\label{propertyV8}
Let $x,y\in Z_{2^a},a\in\{0,1,2,3\}, x\neq y$. % Identify $Z_{2^a}$ with $\mathbb{F}_2^a$.
 Then $\vert (\psi_{2^a}(x)\oplus \psi_{2^a}(y))\cdot(x\oplus y)\vert$ is an odd integer. 
\end{lemma}
%%%%%%%%%%%%%%%%%%%%%
\begin{proof}
It can be proved easily by direct check.
\end{proof}

\begin{lemma}
\label{propertyphi}
% Let $F$ be the set given in \eqref{defineef} and 
Let $x,y\in F, x\neq y$. Then \\
(i) $\vert \overline{\phi_2(x)}\cdot x\vert$  is odd for all $x\neq 0$.\\
(ii) $\vert \overline{\phi_2(x)}\cdot y\vert +
\vert \overline{\phi_2(y)}\cdot x\vert$ is  odd for all  $x\neq y,
x\neq 0,y\neq 0$.\\
\end{lemma}
%%%%%%%%%%
\begin{proof}
%This can be proved easily by 
There are only finitely many possibilities for $x$ and $y$ and it can be easily checked that both the statements (i) and (ii) hold for all possible cases.
\end{proof}
We now have the following important theorem.
%%%%%%%%%%%%%%%%%%%
\begin{theorem}
\label{orthog}
Let $t$ be a positive integer which is a power of $2$. Let $\psi_t$ and $\hat{Z}_{\rho(t)}$ be as defined above. Then,
$ \vert (\psi_t(x)\oplus \psi_t(y))\cdot(x\oplus y)\vert$ is odd for all $x,y\in\hat{Z}_{\rho(t)}, x\neq y$.
\end{theorem}
%%%%%%%%%%%%%
\begin{proof}
For $t=1,2,4$ and $8$, the statement holds by Lemma \ref{propertyV8}. Hence we assume that
$t\geq 16$.
As $\psi_t(0)=0$, it is enough to prove that \\
(i) $\vert \psi_t(y)\cdot y\vert$ is odd for all $y\neq 0$.\\
(ii) $\vert\psi_t(x)\cdot y\vert +\vert\psi_t(y)\cdot x\vert$ is odd for all  $x\neq y,~~x\neq 0,y\neq 0$.\\
To prove (i),
let $z =\psi_t(y)\cdot y$. If $y\in Z_8$, we have $\vert\psi_t(y)\cdot y\vert=\vert\psi_8(y)\cdot y\vert$ which is an odd number by Lemma \ref{propertyV8}.\\
On the other hand, if $y=2^{4l-1}m$, $l\geq 0,m\in F$, then
$\vert z \vert=\vert\overline{2^{4l-1}\phi_2(m)}\cdot 2^{4l-1}m\vert$ where the $2's$ complement of an element is performed in $\mathbb{F}_2^a$.
We have $\vert z \vert=\vert\overline{\phi_2(m)}\cdot m\vert$ where the $2's$ complement of $\phi_2(m)$ is performed in $\mathbb{F}_2^4$. Hence $\vert z \vert$ is odd by Lemma \ref{propertyphi}.

In order to prove the part (ii),
we have following three cases: \\
(i) $1\leq x\leq 7 ~\&~ 1\leq y\leq 7$, \\
(ii) $1\leq y\leq 7~\&~ x=2^{4\alpha-1}\beta$ for some $\beta\in F$, $\alpha\geq 1$,\\
(iii) $x=2^{4\hat{\alpha}-1}\hat{\beta}~\&~ y=2^{4\alpha-1}\beta$ for some $\beta,\hat{\beta}\in F,\alpha,\hat{\alpha}\geq 1$.\\
In all the three cases, we have $x\neq y$.
By Lemma \ref{propertyV8}, (i) is true.  \\ %  follows from the.\\
For the second case, let $z=\psi_t(x)\cdot y \oplus \psi_t(y)\cdot x$.
We have $z=(\overline{2^{4\alpha-1}\phi_2(\beta)}\cdot y) \oplus ((2^{4\alpha-1}\beta)\cdot\overline{\phi_1(y)})$.

As $\overline{2^{4\alpha-1}\phi_2(\beta)}\cdot y=\mathbf{0}$ (the all zero vector in $\mathbb{F}_2^{a}$) for $\alpha\geq 1$, we have
$z=(2^{4\alpha-1}\beta)\cdot\overline{\phi_1(y)}$.
But $\vert\beta\vert$ is odd  for all $\beta\in F$, hence $\vert z \vert$ is an odd number.

For (iii), let $z=\psi_t(x)\cdot y \oplus \psi_t(y)\cdot x$.
We have
\[z=\overline{2^{4\alpha-1}\phi_2(\beta)}\cdot 2^{4\hat{\alpha}-1}\hat{\beta}
\oplus 2^{4\alpha-1}\beta\cdot \overline{2^{4\hat{\alpha}-1}\phi_2(\hat{\beta})}.
\]
If $\hat{\alpha}>\alpha$, we have $2^{4\alpha-1}\beta\cdot \overline{2^{4\hat{\alpha}-1}\phi_2(\hat{\beta})}=\mathbf{0}$ and
$\overline{2^{4\alpha-1}\phi_2(\beta)}\cdot 2^{4\hat{\alpha}-1}\hat{\beta}=\hat{\beta}$. Thus $\vert z\vert$ is an odd number by Lemma \ref{propertyphi}.
If $\alpha=\hat{\alpha}$, it follows that
$\vert z\vert = \vert \overline{\phi_2(\beta)}\cdot \hat{\beta}\vert +
\vert \beta\cdot \overline{\phi_2(\hat{\beta})}\vert$

is an odd number by Lemma \ref{propertyphi}. 
\end{proof}

%%%%%%%%%%%%%%%%%%%%%%%%
The square ROD obtained using the maps $\gamma_t$ and $\psi_t$ given by \eqref{gammat} and \eqref{mappsi} respectively will be denoted by $R_t$ throughout. The RODs $R_{16}$
and $R_{32}$ are given by \eqref{R16} and \eqref{R32} respectively.
In Appendix \ref{appendixI}, it is shown that the RODs $R_t$ can be constructed recursively.
%%%%%%%%%%%%%%%%%%%%%%%
\begin{figure*}
\normalsize{
\begin{eqnarray}
\label{W9}
W_9=\left[ \hspace{-5pt}
\begin{array}{  r @{\hspace{.2pt}} r @{\hspace{.2pt}} r @{\hspace{.2pt}} r @{\hspace{.2pt}} r @{\hspace{.2pt}} r @{\hspace{.2pt}} r @{\hspace{.2pt}} r @{\hspace{.2pt}} r}
   y_0 &  y_1 &  y_2 &  y_3 &  y_4 &  y_5 &  y_6 &  y_7 &  y_8 \\
   y_1 & -y_0 &  y_3 & -y_2 &  y_5 & -y_4 & -y_7 &  y_6 &  y_9 \\
   y_2 & -y_3 & -y_0 &  y_1 &  y_6 &  y_7 & -y_4 & -y_5 & y_{10} \\
   y_3 &  y_2 & -y_1 & -y_0 &  y_7 & -y_6 &  y_5 & -y_4 & y_{11} \\
   y_4 & -y_5 & -y_6 & -y_7 & -y_0 &  y_1 &  y_2 &  y_3 & y_{12} \\
   y_5 &  y_4 & -y_7 &  y_6 & -y_1 & -y_0 & -y_3 &  y_2 & y_{13} \\
   y_6 &  y_7 &  y_4 & -y_5 & -y_2 &  y_3 & -y_0 & -y_1 & y_{14} \\
   y_7 & -y_6 &  y_5 &  y_4 & -y_3 & -y_2 &  y_1 & -y_0 & y_{15} \\
   y_8 & -y_9 & -y_{10} & -y_{11} & -y_{12} & -y_{13} & -y_{14} & -y_{15} & -y_0 \\
   y_9 &  y_8 & -y_{11} & y_{10} & -y_{13} & y_{12} & y_{15} & -y_{14} & -y_1 \\
  y_{10} & y_{11} &  y_8 & -y_9 & -y_{14} & -y_{15} & y_{12} & y_{13} & -y_2 \\
  y_{11} & -y_{10}&  y_9 &  y_8 & -y_{15} & y_{14} & -y_{13} & y_{12} & -y_3 \\
  y_{12} & y_{13} & y_{14} & y_{15} &  y_8 & -y_9 & -y_{10} & -y_{11} & -y_4 \\
  y_{13} & -y_{12}& y_{15} & -y_{14} &  y_9 &  y_8 & y_{11} & -y_{10} & -y_5 \\
  y_{14} & -y_{15}& -y_{12} & y_{13} & y_{10} & -y_{11} &  y_8 &  y_9 & -y_6 \\
  y_{15} & y_{14} & -y_{13} & -y_{12} & y_{11} & y_{10} & -y_9 &  y_8 & -y_7
\end{array}\right], ~~
\hat{W}_9=\left[ \hspace{-5pt}
\begin{array}{ r @{\hspace{.2pt}} r @{\hspace{.2pt}} r @{\hspace{.2pt}} r @{\hspace{.2pt}} r @{\hspace{.2pt}} r @{\hspace{.2pt}} r @{\hspace{.2pt}} r @{\hspace{.2pt}} r}
   y_0 &    -y_1 &    -y_2 &     y_3 &    -y_4 &     y_5 &     y_6 &    -y_7 &    -y_8\\
   y_1 &     y_0 &    -y_3 &    -y_2 &    -y_5 &    -y_4 &    -y_7 &    -y_6 &    -y_9\\
   y_2 &     y_3 &     y_0 &     y_1 &    -y_6 &     y_7 &    -y_4 &     y_5 &   -y_{10}\\
   y_3 &    -y_2 &     y_1 &    -y_0 &    -y_7 &    -y_6 &     y_5 &     y_4 &   -y_{11}\\
   y_4 &     y_5 &     y_6 &    -y_7 &     y_0 &     y_1 &     y_2 &    -y_3 &   -y_{12}\\
   y_5 &    -y_4 &     y_7 &     y_6 &     y_1 &    -y_0 &    -y_3 &    -y_2 &   -y_{13}\\
   y_6 &    -y_7 &    -y_4 &    -y_5 &     y_2 &     y_3 &    -y_0 &     y_1 &   -y_{14}\\
   y_7 &     y_6 &    -y_5 &     y_4 &     y_3 &    -y_2 &     y_1 &     y_0 &   -y_{15}\\
   y_8 &     y_9 &    y_{10} &   -y_{11} &    y_{12} &   -y_{13} &   -y_{14} &    y_{15} &    y_0\\
   y_9 &    -y_8 &    y_{11} &    y_{10} &    y_{13} &    y_{12} &    y_{15} &    y_{14} &    y_1\\
  y_{10} &   -y_{11} &    -y_8 &    -y_9 &    y_{14} &   -y_{15} &    y_{12} &   -y_{13} &    y_2\\
  y_{11} &    y_{10} &    -y_9 &     y_8 &    y_{15} &    y_{14} &   -y_{13} &   -y_{12} &    y_3\\
  y_{12} &   -y_{13} &   -y_{14} &    y_{15} &    -y_8 &    -y_9 &   -y_{10} &    y_{11} &    y_4\\
  y_{13} &    y_{12} &   -y_{15} &   -y_{14} &    -y_9 &     y_8 &    y_{11} &    y_{10} &    y_5\\
  y_{14} &    y_{15} &    y_{12} &    y_{13} &   -y_{10} &   -y_{11} &     y_8 &    -y_9 &    y_6\\
  y_{15} &   -y_{14} &    y_{13} &   -y_{12} &   -y_{11} &    y_{10} &    -y_9 &    -y_8 &    y_7\\
\end{array}\right]
\end{eqnarray}
}
\hrule
\end{figure*}

%%%%%%%%%%%%%%%%%%%%%%%%%%%%%%%
One can define the functions $\gamma_t$ and $\psi_t$ different from the one given above and can have a square ROD different from $R_t$. In Appendix \ref{appendixII}, we provide three different pairs of such functions  and these are shown to give the well-known  Adams-Lax-Phillips' construction from Octonions and Quaternions and Geramita and Pullman's construction of square RODs.
%%%%%%%%%%%%%%%%%%%%%%%%%%
\subsection{STEP 2 : Construction of new sets of rate-1 RODs}
\label{subsec2-2}
Transition from a square ROD to a rate-1 ROD can be performed using column vector representation of an ROD~\cite{TJC}. In a similar way, we construct a rate-1 ROD $W_n$ of size $[\nu(n),n,\nu(n)]$ for $n$ transmit antennas from an ROD of size $[\nu(n),\nu(n),n]$ where $n$ is any non-zero positive integer, not necessarily a power of 2. 

Any square ROD of order $\nu(n)$ obtained via a suitable pair of maps $\gamma_{\nu(n)}$ and $\psi_{\nu(n)}$ satisfying the condition \eqref{oddcondition} % of Theorem \ref{ratesquarerod} 
(for instance, $R_{\nu(n)}$ obtained in the previous subsection or  $A_{\nu(n)},\hat{A}_{\nu(n)}$ and $G_{\nu(n)}$ obtained in Appendix \ref{appendixII}) can be used for this purpose. We refer to any such design by $B_{\nu(n)}$ consisting of $n$ real variables. % $z_0,z_1,\cdots,z_{n-1}$.

Let $y_0,y_1,\cdots,y_{\nu(n)-1}$ be $\nu(n)$ real variables. % which constitute the matrix $W_n$.
The matrix $W_n$ is obtained as follows: Make
$W_n(i,j)=0$ if the $i$-th row of $B_{\nu(n)}$  does not contain $z_j$. Otherwise,
$W_n(i,j)=y_k$ or $-y_k$ if $B_{\nu(n)}(i,k)=z_j$ or $-z_j$ respectively.
The construction of the matrix $W_n$ ensures that it is a rate-1 ROD. 
Using Theorem \ref{ratesquarerod} and Theorem \ref{orthog}, we have
\begin{align}
\label{defineW}
%\begin{array}{l}
 &W_n(i,j)=s(i,j)y_{f(i,j)} \mbox{ where } \nonumber\\
 &f(i,j)=i\oplus \gamma_{\nu(n)}(j) % \nonumber\\
,s(i,j)=(-1)^{\vert i\cdot\psi_{\nu(n)}(\gamma_{\nu(n)}(j))\vert}
%\end{array}
\end{align}
for $0\leq i\leq \nu(n)-1,~ 0\leq j\leq n-1$.
Similarly, we define another matrix $\hat{W}_n$ as
\begin{align}
\label{defineWhat}
%\begin{array}{l}
&\hat{W}_n(i,j)=\hat{s}(i,j)y_{f(i,j)} \mbox{ where } 
f(i,j)=i\oplus \gamma_{\nu(n)}(j),
\nonumber\\
&\hat{s}(i,j)=(-1)^{\vert (i\oplus \gamma_{\nu(n)}(j))\cdot\psi_{\nu(n)}(\gamma_{\nu(n)}(j))\vert}.
%\end{array}
\end{align}
$\hat{W}_n$ is also a rate-1 ROD. $\hat{W}_n$ and $W_n$ are used to construct a rate-$\frac{1}{2}$ scaled-COD for $(n+8)$ antennas.
Two rate-1 RODs $W_9$ and $\hat{W}_9$ for $9$ antennas are given by \eqref{W9}.
%%%%%%%%%%%%%%%%%%%%%%%%%%
%%%%%%%%%%%%%%%%%%%%%%%%%%%%%%%%%%%%%%%%%%%%%%%%
\subsection{STEP 3 : Construction of low-delay, rate-$\frac{1}{2}$ scaled-CODs}
\label{subsec2-3}
The construction of the rate-$\frac{1}{2}$ code is little involved: it makes use of two rate-1 RODs constructed in the previous subsection and the code-matrix contains several copies of square COD of size $[8,8,4]$.
For $n$ transmit antennas, the desired rate-$\frac{1}{2}$ scaled-COD $RH_n$ is given by  
\begin{equation}
\label{gn}
   RH_n=\left[\begin{array}{cc}
     E_8 & H_t\\
     O_8 & \hat{H}_t 
\end{array}\right]
\end{equation}
where $t=n-8$. The matrices $E_8,H_t,O_8$ and $\hat{H}_t$ are constructed as follows.
$H_t$ and $\hat{H}_t$ are constructed very easily using rate-1 RODs and 
an $8\times 1$ column vector given by

{\small
\begin{eqnarray*}
C(x_0,x_1,x_2,x_3) 
=\frac{1}{\sqrt{2}}\left[
%\begin{array} {cccccccccc}
% \begin{array}{ r @{\hspace{.8pt}} r @{\hspace{.8pt}} r @{\hspace{.8pt}} r @{\hspace{.8pt}} r @{\hspace{.8pt}} r @{\hspace{.8pt}} r @{\hspace{.8pt}} r @{\hspace{.8pt}} }
\begin{array}{ c @{\hspace{2pt}} c @{\hspace{2pt}} c @{\hspace{2pt}} c @{\hspace{2pt}} c @{\hspace{2pt}} c @{\hspace{2pt}} c @{\hspace{2pt}} c }
-x_3^* & x_2^* &-x_1^* &-x_0 & x_0^* &-x_1 &-x_2  &-x_3
\end{array}\right]^\mathcal{T}
\end{eqnarray*}
}

\noindent
where $x_0,x_1,\cdots$ are complex variables.
Define 
$\overline{A}(i)=C(x_{4i},x_{4i+1},x_{4i+2},x_{4i+3})$
for all non-negative integer $i$.

Let $W_t$ and $\hat{W}_t$ be two rate-1 RODs of size $[\nu(t),t,\nu(t)]$ in $\nu(t)$ real variables $y_0,y_1,\cdots,y_{\nu(t)-1}$ as constructed in the previous subsection.
Let $H_t$ be the matrix obtained from $W_t$ by substituting $y_i$ with $\overline{A}(2i+1)$ for $i=0$ to $\nu(t)-1$.
Similarly construct $\hat{H}_t$ from $\hat{W}_t$ by substituting $y_i$ with $\overline{A}(2i)$

Next, we construct $E_8$ and $O_8$.
Let
{\small
\begin{equation}
\label{tx8cod1}
   A(x_0,x_1,x_1,x_3)=\left[
\begin{array}{ r @{\hspace{.2pt}} r @{\hspace{.2pt}} r @{\hspace{.2pt}} r @{\hspace{.2pt}} r @{\hspace{.2pt}} r @{\hspace{.2pt}} r @{\hspace{.2pt}} r @{\hspace{.2pt}} }
     x_0 & -x_1^* & -x_2^* & 0   & -x_3^* & 0 & 0 & 0 \\
     x_1 & x_0^*  & 0  & -x_2^* & 0 & -x_3^* & 0 & 0  \\
     x_2 & 0 & x_0^* & x_1^*  & 0 & 0 & -x_3^* & 0 \\
     0   & x_2 & -x_1 & x_0  & 0 & 0 & 0 & -x_3^*  \\
     x_3 & 0 & 0 & 0 & x_0^* & x_1^* &x_2^* &0    \\
     0   & x_3 & 0 & 0  & -x_1 & x_0 &0 &x_2^*    \\
     0   &0  & x_3 & 0  & -x_2 &0 & x_0 & -x_1^* \\
     0   &0  &0 & x_3 & 0  & -x_2  & x_1 & x_0^* 
\end{array}\right],
\end{equation}
}
{\small
\begin{equation}
   B(x_4,x_5,x_6,x_7)=\left[%\begin{array}{cccccccccc}
\begin{array}{ r @{\hspace{.2pt}} r @{\hspace{.2pt}} r @{\hspace{.2pt}} r @{\hspace{.2pt}} r @{\hspace{.2pt}} r @{\hspace{.2pt}} r @{\hspace{.2pt}} r @{\hspace{.2pt}} }
    x_4 &-x_5^* &-x_6^* &-x_7^*  & 0 & 0 & 0 & 0 \\
    x_5 & x_4^* & 0 & 0  & -x_6^* &-x_7^* & 0 & 0 \\
    x_6 &0  & x_4^* & 0  & x_5^* &0 &-x_7^* & 0 \\
     0  & x_6 & -x_5 &0  & x_4  & 0  & 0 &-x_7^* \\
     x_7&0 &0 & x_4^*  &0 & x_5^*   & x_6^*  & 0 \\
     0  & x_7  & 0  & -x_5 & 0 & x_4 & 0 & x_6^* \\
     0  & 0 & x_7 &-x_6  & 0 & 0 & x_4 &-x_5^* \\
     0  & 0 & 0 & 0 & x_7 &-x_6 & x_5   & x_4^*
    \end{array}\right]
\end{equation}
}
be two square CODs of size $[8,8,4].$ Define
\begin{eqnarray*}
A(2i)&=&A(x_{8i},x_{8i+1},x_{8i+2},x_{8i+3})\\
A(2i+1)&=&B(x_{8i+4},x_{8i+5},x_{8i+6},x_{8i+7}). 
\end{eqnarray*}
We now construct two $\frac{\nu(n)}{2}\times 8$ matrices $E_8$ and $O_8$ using $A(i)$ as follows:
\begin{eqnarray}
\label{rowvectorrep}
E_8= 
\left[
\begin{array}{c}
A(0) \\
A(2) \\
. \\
. \\
. \\
A(u-2)
\end{array}
\right],~~
O_8=
\left[
\begin{array}{c}
A(1) \\
A(3) \\
. \\
. \\
. \\
A(u-1)
\end{array}
\right]
\end{eqnarray}
~~~~where $u=\nu(n)/8$. 
%Note that each is a $4u\times 8$ matrix.
% One can easily verify that 
Note that
\begin{eqnarray}
\label{Aij}
  \left[\begin{array}{cc}
    A(i) &\overline{A}(j)\\
    A(j) &\overline{A}(i)
   \end{array}\right]
\end{eqnarray}
is a scaled-COD whenever $(i+j)$ is odd and
\begin{eqnarray}
\label{Aoverline}
\left[\begin{array}{cc}
    \overline{A}(i) &-\overline{A}(j)\\
    \overline{A}(j) &\overline{A}(i)
   \end{array}\right],~~
\end{eqnarray}
is a scaled-COD for all values of $i$ and $j$, $i\neq j$.

\begin{figure*}
{\footnotesize
\begin{eqnarray}
\label{ratehalf_9}
   \left[\begin{array}{rrrrrrrrr} % ccccccccc
     x_0 & -x_1^* & -x_2^* & 0   & -x_3^* & 0 & 0 & 0 &\frac{-x_7^*}{\sqrt{2}} \\
     x_1 & x_0^*  & 0  & -x_2^* & 0 & -x_3^* & 0 & 0  &\frac{x_6^*}{\sqrt{2}}  \\
     x_2 & 0 & x_0^* & x_1^*  & 0 & 0 & -x_3^* & 0 &\frac{-x_5^*}{\sqrt{2}} \\
     0   & x_2 & -x_1 & x_0  & 0 & 0 & 0 & -x_3^*  &\frac{-x_4}{\sqrt{2}} \\
     x_3 & 0 & 0 & 0 & x_0^* & x_1^* &x_2^* &0    &\frac{x_4^*}{\sqrt{2}}   \\
     0   & x_3 & 0 & 0  & -x_1 & x_0 &0 &x_2^*     &\frac{-x_5}{\sqrt{2}} \\
     0   &0  & x_3 & 0  & -x_2 &0 & x_0 & -x_1^* &\frac{-x_6}{\sqrt{2}} \\
     0   &0  &0 & x_3 & 0  & -x_2  & x_1 & x_0^* &\frac{-x_7}{\sqrt{2}} \\
     x_4 &-x_5^* &-x_6^* &-x_7^*  & 0 & 0 & 0 & 0  &\frac{-x_3^*}{\sqrt{2}}\\
     x_5& x_4^* & 0 & 0  & -x_6^* &-x_7^* & 0 & 0   &\frac{x_2^*}{\sqrt{2}} \\
    x_6 &0  & x_4^* & 0  & x_5^* &0 &-x_7^* & 0 &\frac{-x_1^*}{\sqrt{2}}  \\
     0    & x_6 & -x_5 &0  & x_4  & 0  & 0 &-x_7^* &\frac{-x_0}{\sqrt{2}}\\
     x_7 &0 &0 & x_4^*  &0 & x_5^*   &-x_7^*  & 0  &\frac{x_0^*}{\sqrt{2}} \\
     0    & x_7  & 0  & -x_5 & 0 & x_4 & 0 & x_6^*  &\frac{-x_1}{\sqrt{2}}  \\
     0    & 0 & x_7 &-x_6  & 0 & 0 & x_4 &-x_5^* &\frac{-x_2}{\sqrt{2}}   \\
     0    & 0 & 0 & 0 & x_7 &-x_6 & x_5   & x_4^*  &\frac{-x_3}{\sqrt{2}}\\
\end{array}\right],~~
\frac{1}{\sqrt{2}}
\left[\begin{array}{r@{\hspace{0.2cm}}r@{\hspace{0.2cm}}r@{\hspace{0.2cm}}r@{\hspace{0.2cm}}r@{\hspace{0.2cm}}r@{\hspace{0.2cm}}r@{\hspace{0.2cm}}r@{\hspace{0.2cm}}r@{\hspace{0.2cm}}}
    x_0   &-x_1   &-x_2    &-x_3    &-x_4    &-x_5     &-x_6  &-x_7 &-x_8\\
    x_1   & x_0   & x_3    &-x_2    & x_5    &-x_4     &-x_7  & x_6 & x_9\\
    x_2   &-x_3   & x_0    & x_1    & x_6    &x_7      &-x_4  &-x_5 & x_{10}\\
    x_3   & x_2   &-x_1    & x_0    & x_7    &-x_6     & x_5  &-x_4 & x_{11}\\
    x_4   &-x_5   &-x_6    &-x_7    & x_0    & x_1     & x_2  & x_3 & x_{12}\\
    x_5   & x_4   &-x_7    & x_6    &-x_1    & x_0     &-x_3  & x_2 & x_{13}  \\
    x_6   & x_7   & x_4    &-x_5    &-x_2    &x_3      & x_0  &-x_1 & x_{14}  \\
    x_7   &-x_6   & x_5    & x_4    &-x_3    &-x_2     & x_1  & x_0 & x_{15}\\
    x_8&-x_9&-x_{10} &-x_{11} &-x_{12} &-x_{13}  &-x_{14}&-x_{15} & x_0\\
    x_9& x_8   &-x_{11} & x_{10} &-x_{13} & x_{12}  & x_{15}&-x_{14}&-x_1\\
    x_{10}& x_{11}& x_8    &-x_9 &-x_{14} &-x_{15}  & x_{12}& x_{13}&-x_2\\
    x_{11}&-x_{10}& x_9 & x_8    &-x_{15} & x_{14}  &-x_{13}& x_{12}&-x_3\\
    x_{12}& x_{13}& x_{14} & x_{15} & x_8    &-x_9  &-x_{10}&-x_{11}&-x_4\\
    x_{13}&-x_{12}& x_{15} &-x_{14} & x_9 & x_8     & x_{11}&-x_{10}&-x_5   \\
    x_{14}&-x_{15}&-x_{12} & x_{13} & x_{10} &-x_{11}  & x_8   & x_9&-x_6\\
    x_{15}&x_{14} &-x_{13} &-x_{12} & x_{11} & x_{10}  &-x_9& x_8   &-x_7\\
    x^*_0   &-x^*_1   &-x^*_2    &-x^*_3    &-x^*_4    &-x^*_5     &-x^*_6  &-x^*_7 &-x^*_8\\
    x^*_1   & x^*_0   & x^*_3    &-x^*_2    & x^*_5    &-x^*_4     &-x^*_7  & x^*_6 & x^*_9\\
    x^*_2   &-x^*_3   & x^*_0    & x^*_1    & x^*_6    &x^*_7      &-x^*_4  &-x^*_5 & x^*_{10}\\
    x^*_3   & x^*_2   &-x^*_1    & x^*_0    & x^*_7    &-x^*_6     & x^*_5  &-x^*_4 & x^*_{11}\\
    x^*_4   &-x^*_5   &-x^*_6    &-x^*_7    & x^*_0    & x^*_1     & x^*_2  & x^*_3 & x^*_{12}\\
    x^*_5   & x^*_4   &-x^*_7    & x^*_6    &-x^*_1    & x^*_0     &-x^*_3  & x^*_2 & x^*_{13}  \\
    x^*_6   & x^*_7   & x^*_4    &-x^*_5    &-x^*_2    &x^*_3      & x^*_0  &-x^*_1 & x^*_{14}  \\
    x^*_7   &-x^*_6   & x^*_5    & x^*_4    &-x^*_3    &-x^*_2     & x^*_1  & x^*_0 & x^*_{15}\\
    x^*_8&-x^*_9&-x^*_{10} &-x^*_{11} &-x^*_{12} &-x^*_{13}  &-x^*_{14}&-x^*_{15} & x^*_0\\
    x^*_9& x^*_8   &-x^*_{11} & x^*_{10} &-x^*_{13} & x^*_{12}  & x^*_{15}&-x^*_{14}&-x^*_1\\
    x^*_{10}& x^*_{11}& x^*_8    &-x^*_9 &-x^*_{14} &-x^*_{15}  & x^*_{12}& x^*_{13}&-x^*_2\\
    x^*_{11}&-x^*_{10}& x^*_9 & x^*_8    &-x^*_{15} & x^*_{14}  &-x^*_{13}& x^*_{12}&-x^*_3\\
    x^*_{12}& x^*_{13}& x^*_{14} & x^*_{15} & x^*_8    &-x^*_9  &-x^*_{10}&-x^*_{11}&-x^*_4\\
    x^*_{13}&-x^*_{12}& x^*_{15} &-x^*_{14} & x^*_9 & x^*_8     & x^*_{11}&-x^*_{10}&-x^*_5   \\
    x^*_{14}&-x^*_{15}&-x^*_{12} & x^*_{13} & x^*_{10} &-x^*_{11}  & x^*_8   & x^*_9&-x^*_6\\
    x^*_{15}&x^*_{14} &-x^*_{13} &-x^*_{12} & x^*_{11} & x^*_{10}  &-x^*_9& x^*_8   &-x^*_7
\end{array}\right]
\end{eqnarray}
}
\hrule
\end{figure*}

Note that the number of rows and columns of the matrix $RH_n$ are $16\cdot \nu(n-8)=8\cdot \nu(n)/8 =\nu(n)$ and $t+8=n$ respectively.
The following theorem is the main result of this paper.
%%%%%%%%%%%%%%%%%%%%%%%%%
\begin{theorem}
\label{rate12cod}
For any non-zero positive integer $n$, there exists a 
rate-$\frac{1}{2}$ scaled-COD for $n$ transmit antennas with decoding delay $\nu(n)$.
\end{theorem} 

%%%%%%%%%%%%%%%%%%%%%
\begin{proof}
For $n\leq 8$, one can construct rate-$\frac{1}{2}$ COD of size $[\nu(n),n,\frac{\nu(n)}{2}]$
from  a COD of size $[8,8,4]$ given by \eqref{tx8cod1}. We assume that $n\geq 9$. We claim that the matrix $RH_n$ given by \eqref{gn} is a rate-$\frac{1}{2}$ scaled-COD for $n$ transmit antennas with decoding delay $\nu(n)$. 

Let $p=\nu(n)$.
We have
\begin{equation*}
RH_n^\mathcal{H}RH_n=\left[\begin{array}{cc}
     E_8^\mathcal{H}E_8+O_8^\mathcal{H}O_8  & E_8^\mathcal{H}H_t + O_8^\mathcal{H}\hat{H}_t\\
     H_t^\mathcal{H}E_8 + \hat{H}_t^\mathcal{H}O_8 & H_t^\mathcal{H}H_t+\hat{H}_t^\mathcal{H}\hat{H}_t
\end{array}\right].
\end{equation*}
From the construction of $E_8$ and $O_8$ given by \eqref{rowvectorrep},
we have \\
 $E_8^\mathcal{H}E_8+O_8^\mathcal{H}O_8= ({\vert x_0\vert}^2 +\cdots+{\vert x_{\frac{p}{2}-1}\vert}^2)I_8$.
From equation \eqref{Aoverline}, we have
\begin{eqnarray*}
H_t^\mathcal{H}H_t+\hat{H}_t^\mathcal{H}\hat{H}_t= ({\vert x_0\vert}^2 +\cdots+{\vert x_{p/2-1}\vert}^2)I_{n-8}.
\end{eqnarray*}
Thus it is enough to prove that $E_8^\mathcal{H}H_t + O_8^\mathcal{H}\hat{H}_t=0_{8\times (n-8)}$ where $0_{8\times (n-8)}$ is a matrix of size $8\times (n-8)$ containing zero only.
Let the $j$-th column of $H_t$ and $\hat{H}_t$ be $H_t(j)$ and $\hat{H}_t(j)$ respectively.
Then we show that $Z(j)=E_8^\mathcal{H}H_t(j) + O_8^\mathcal{H}\hat{H}_t(j)=0_{8\times 1}$ for all $j\in\{0,1,\cdots,n-8-1\}$.\\
Let $u=p/8$. For convenience, we write $\gamma$ for $\gamma_{\nu(t)}$.
We have
\begin{eqnarray*}
\begin{array}{c}
E_8^\mathcal{H}= \left[
\begin{array}{cccc}
A^\mathcal{H}(0) & A^\mathcal{H}(2)&\cdots & A^\mathcal{H}(u-2)
\end{array}
\right],\\
\vspace*{.40cm}
O_8^\mathcal{H}= \left[
\begin{array}{cccc}
A^\mathcal{H}(1) & A^\mathcal{H}(3)&\cdots & A^\mathcal{H}(u-1)
\end{array}
\right],
\end{array}
\end{eqnarray*}
%\end{eqnarray*}
\begin{eqnarray*}
%\vspace*{.40cm}
%\begin{array}{c}
H_t(j)&=
\left[
\begin{array}{c}
s(0,j)\overline{A}({2(0\oplus \gamma(j))+1}) \\
s(1,j)\overline{A}({2(1\oplus \gamma(j))+1}) \\
. \\
. \\
s(i,j)\overline{A}({2(i\oplus \gamma(j))+1}) \\
. \\
. \\
s(\frac{u}{2}-1,j)\overline{A}({2\bigl((\frac{u}{2}-1)\oplus \gamma(j)\bigr)+1})
\end{array}
\right], \\
\vspace*{.80cm}
\hat{H}_t(j)&=
\left[
\begin{array}{c}
\hat{s}(0,j)\overline{A}({2(0\oplus \gamma(j))}) \\
\hat{s}(1,j)\overline{A}({2(1\oplus \gamma(j))}) \\
. \\
. \\
\hat{s}(i,j)\overline{A}({2(i\oplus \gamma(j))}) \\
. \\
. \\
\hat{s}(\frac{u}{2}-1,j)\overline{A}({2((\frac{u}{2}-1)\oplus \gamma(j))})
\end{array}
\right],
%\end{array}
\end{eqnarray*}
where $s(i,j)$ and $\hat{s}(i,j)$ are given by \eqref{defineW} and \eqref{defineWhat} respectively. We have
\begin{eqnarray*}
Z(j)&=&\sum_{i=0}^{\frac{u}{2}-1}s(i,j)A^\mathcal{H}(2i)\overline{A}({2(i\oplus \gamma(j))+1})\\
 &&+ \sum_{i=0}^{\frac{u}{2}-1}\hat{s}(i,j)A^\mathcal{H}(2i+1)\overline{A}({2(i\oplus \gamma(j))}).
\end{eqnarray*}
Now $s(i,j)=\hat{s}(i\oplus \gamma(j),j)$ and 
\begin{align*}
%\begin{array}{l}
&\sum_{i=0}^{\frac{u}{2}-1}\hat{s}(i,j)A^\mathcal{H}(2i+1)\overline{A}({2(i\oplus \gamma(j))})\\
&=\sum_{i=0}^{\frac{u}{2}-1}\hat{s}(i\oplus \gamma(j),j)A^\mathcal{H}(2(i\oplus \gamma(j))+1)\overline{A}({2i})\bigr).
%\end{array}
% \mbox {and } s(i,j)=\hat{s}(i\oplus \gamma(j),j).
\end{align*}
Therefore,
\begin{eqnarray*}
Z(j)&=&\sum_{i=0}^{\frac{u}{2}-1}\bigl(s(i,j)A^\mathcal{H}(2i)\overline{A}({2(i\oplus \gamma(j))+1})\\  &&+ \hat{s}(i\oplus \gamma(j),j)A^\mathcal{H}(2(i\oplus \gamma(j))+1)\overline{A}({2i})\bigr)\\
&=&\sum_{i=0}^{\frac{u}{2}-1}s(i,j)(A^\mathcal{H}(2i)\overline{A}({2(i\oplus \gamma(j))+1})\\
  &&+ A^\mathcal{H}(2(i\oplus \gamma(j))+1)\overline{A}({2i}))\\
 &=&0_{8 \times 1} %\mbox{ using \eqref{Aij}.}
\end{eqnarray*}
as the matrix given by \eqref{Aij} is a scaled-COD.
\end{proof}
%%%%%%%%%%%%%%%%
\noindent
% We illustrate our main result in the following example. 
\begin{example}
For $9$ transmit antennas, the rate-$\frac{1}{2}$ scaled-COD of size $[16,9,8]$ and 
the known rate-$\frac{1}{2}$ scaled-COD~\cite{TJC} of size $[32,9,16]$ are given by \eqref{ratehalf_9}.
For $10$ transmit antennas, the proposed rate-$\frac{1}{2}$ code of size $[32,10,16]$ is given in Appendix \ref{appendixIII}. %by \eqref{new321016}.
\end{example}
%%%%%%%%%%%%%%%%%%%%%%%%%%%%%%%%%%%%%%%%%%

It has been shown by Liang~\cite{Lia} that the maximal rate of a COD for $n$ transmit antennas is $\frac{1}{2}+\frac{1}{2t}$ when $n=2t-1$ or $2t$. 
However, the rate of a scaled-COD, with scaling of at least one column is at most half as
each variable appears twice in that column and therefore $k/p\leq 1/2$ where 
$k$ is the number of complex variables and $p$ is the number of rows of the design.

\subsection{Summary of the proposed rate-$\frac{1}{2}$ codes}
It has been observed that the number of complex variables in the proposed rate-$\frac{1}{2}$ code for $n$ transmit antennas is $\frac{\nu(n)}{2}$ and the number of rows is $\nu(n)$ (the number $\nu(n)$ is given by \eqref{nun}).
The construction of these codes requires two rate-1 RODs for $n-8$ antennas. In this paper, we construct $W_{t}$ and $\hat{W_t}$ (where $t=n-8$) given by \eqref{defineW} and \eqref{defineWhat} respectively which are used to construct rate-$\frac{1}{2}$ scaled-CODs  $H_{t}$ and $\hat{H_t}$ (for $t$ transmit antennas) respectively. The matrix 
\begin{eqnarray*}
 \left[\begin{array}{c}
     H_t\\
    \hat{H}_t 
\end{array}\right]
\end{eqnarray*}
constitutes the last $n-8$ columns of the proposed rate-$\frac{1}{2}$ scaled-COD for $n$ antennas while the matrices $E_8$ and $O_8$ given by \eqref{rowvectorrep} constitute the first eight columns of the proposed code.

%%%%%%%%%%%%%%%%%%%%%%%%%%%%%%%%%%%%%%%%%%%%%%%
\section{Delay-minimality for 9 transmit antennas}
\label{sec3}
In this section, it is  shown that the proposed rate-$\frac{1}{2}$ scaled-COD 
for $9$ transmit antennas achieves minimal delay.
To prove this, we need some preliminary facts regarding the interrelationship between ODs and certain bilinear maps. 
It has been observed that~\cite{WaX} the orthogonal designs and bilinear maps are intimately related in the sense that an LPROD of size $[p,n,k]$ exists if and only if there exists a type of bilinear map called {\it{normed bilinear map}} with parameters $p,n$ and $k$. The normed bilinear maps have been studied extensively and one can find a good introduction to this topic in the book by Shapiro~\cite{Sha}. 

A bilinear map $f$ (over a field $\mathbb{F}$) is a map
\begin{eqnarray}
       f : \mathbb{F}^k \times \mathbb{F}^n &\rightarrow& \mathbb{F}^p \\
             (x,y) &\mapsto&    f(x,y)
\end{eqnarray}
\noindent
such that it is linear in both $x$ and $y$, i.e., $f(x_1+x_2,y)=f(x_1,y)+f(x_2,y)$ and $f(x,y_1+y_2)=f(x,y_1)+f(x,y_2)$ for all $x,x_1,x_2\in\mathbb{F}^k$ and $y,y_1,y_2\in\mathbb{F}^n.$ 
If the vector space under consideration is an inner product space, for example, when the field is real numbers or complex numbers, the Euclidean norm of a vector $x$ is denoted by $\norm{x}$.
If a bilinear map preserves the norm, then it is called a  normed bilinear map. More precisely,
\begin{defn}
 A  {\textit{normed real bilinear map}} (NRBM) of size $[p,n,k]$ is a map
 $f : \mathbb{R}^k \times \mathbb{R}^n \rightarrow \mathbb{R}^p$ such that
 $f$ is bilinear and normed i.e., $\norm{f(x,y)} =\norm{x}\norm{y}\forall x\in\mathbb{R}^k, y\in\mathbb{R}^n$.\\
A bilinear map $f$ is called nonsingular if 
$f(x,y)=0$ implies $x=0$ or $y=0$. %, then such a map is called a nonsingular map.
\end{defn}

The following theorem gives a lower bound on $p$ for fixed values of $n$ and $k$.
 \begin{theorem}[Hopf-Stiefel Theorem~\cite{Sha}]
If there exists a nonsingular bilinear map of size $[p,n,k]$ over $\mathbb{R}$, then
      $(x+y)^p =0 \mbox{ in the ring }\mathbb{F}_2[x,y]/(x^n,y^k)$.
\end{theorem}
\begin{defn}
 Let $n,k$ be positive integers. Then the three quantities $n\circ k,$  $p_{BL}(n,k)$ and $p_{NBL}(n,k)$ are defined by
\begin{itemize}
\item $n\circ k  = min \{ p: (x+y)^p = 0 \ \mbox {in} \ {\mathbb{F}_2} [x,y]/(x^n,y^k)\}$,\item    $p_{BL}(n,k) = min \{p: $ there is a nonsingular bilinear map $ [p,n,k] $ over $\mathbb{R}  $ \},
\item   $p_{NBL}(n,k) = min \{p: $ there is a normed bilinear map $ [p,n,k]$ over $ \mathbb{R} $\},
\end{itemize}
\end{defn}
The following  basic facts about these quantities are well-known~\cite{Sha}.\\
 $p_{NBL}(n,k) \geq p_{BL}(n,k) \geq n\circ k$.
It follows from the definition of $ n \circ k $ that
   \begin{prpn}[\cite{Sha}]
   \label{ncirck}
      $ n \circ k $ is a commutative binary operation. \\ % on $ \mathbb{Z}^+ $ .\\
    $(I)$ If $ k \leq l $ then $ n \circ k \leq  n \circ l $ \\
    $(II)$ $ n \circ k  = 2^m $   if and only if $ k,n \leq 2^m $ and $ k + n > 2^m $ .\\
    $(III)$ If $ n \leq 2^m $ then $ n \circ ( k  + 2^m ) = n \circ k + 2^m . $
\end{prpn}
\begin{example}
 To compute $10 \circ 10$, note that $10 < 2^4$, but $(10 +10) > 16$.
Therefore, $10\circ 10 =16$.
\end{example}

The relation between RODs and NRBMs has been observed by Wang and Xia~\cite{WaX}.
The following theorem states that RODs and normed bilinear maps are equivalent.
%%%%%%%%%%%%%%%%%%%%%%%%%%%%%%%%%%%%%%%%
\begin{lemma}
\label{rod_nrbl}
An LPROD of size $[p,n,k]$ exists if and only if there exists a normed real bilinear map of size $[p,n,k]$.
\end{lemma}
%%%%%%

\begin{proof}

Let $\underline{x}\in\R^k$ be the column vector $(x_1,\cdots,x_k)^\mathcal{T}$.\\
 Similarly, define
$\underline{y}=(y_1,\cdots,y_n)^\mathcal{T}$ and
$\underline{z}=(z_1,\cdots,z_p)^\mathcal{T}$.

Let $A$ be an ROD of size $[p,n,k]$ in $k$ variables $x_1,x_2\cdots,x_k$.
Let
\begin{eqnarray*}
f : \mathbb{R}^k \times \mathbb{R}^n &\rightarrow& \mathbb{R}^p\\
(\underline{x},\underline{y})&\mapsto&  A\underline{y}. 
\end{eqnarray*}
The $i$-th row of $A$ is given by $\underline{x}^\mathcal{T}B_i$ where the matrices $B_i,i=1,2,\cdots,p$ are uniquely determined by the matrix $A$.
Let $\underline{z}=f(\underline{x},\underline{y})$. As $z_i=\underline{x}^\mathcal{T}B_i\underline{y}$ for $i=1,2,\cdots,p$, the map $f$ is bilinear.\\
$f$ is normed as $\norm{f(\underline{x},\underline{y})}^2=\norm{ A\underline{y}}^2=(A\underline{y})^\mathcal{T}A\underline{y}=
\underline{y}^\mathcal{T}(x_1^2 + x_2^2 +\cdots+x_k^2)I_n)\underline{y}=
\norm{ \underline{x}}^2 \norm{\underline{y}}^2$.

We now prove the converse. Let $f$ be the normed bilinear map given by 
\begin{eqnarray*}
f : \mathbb{R}^k \times \mathbb{R}^n &\rightarrow& \mathbb{R}^p\\
                                (\underline{x},\underline{y})&\mapsto&   \underline{z}. 
\end{eqnarray*}
As $f$ is linear in both $\underline{x}$ and $\underline{y}$, we have $\underline{z}=A\underline{y}$ where $A$ is a $p\times n$ matrix where each entry of the matrix is a real linear combination of the variables $x_1,\cdots,x_k$.
As $f$ is normed, we have 
$\norm{\underline{z}}^2=\norm{ f(\underline{x},\underline{y})}^2=\norm{\underline{x}}^2 \norm{\underline{y}}^2$.
But $f(\underline{x},\underline{y})=A\underline{y}$.
Then, $\norm{ A\underline{y}}^2=(x_1^2 +\cdots+x_k^2)\underline{y}^\mathcal{T}\underline{y}$ i.e.,
$\underline{y}^\mathcal{T}A^\mathcal{T}A\underline{y}=(x_1^2 +\cdots+x_k^2)\underline{y}^\mathcal{T}\underline{y}$.
As $\underline{y}$ consists of variables, this equation is equivalent to 
$A^\mathcal{T}A=(x_1^2 +\cdots+x_k^2)I_n$.
 \end{proof}

%%%%%%%%%%%%%%%%%%%%%%%%%%%%%%%%%%%
We now prove the main result of this section.
\begin{theorem}
The minimum value of the decoding delay of a rate-$\frac{1}{2}$ LPCOD for $9$ transmit antennas is $16$.

\end{theorem}
%%%%%%%%%%%%%%%%%
\begin{proof}
We prove it by contradiction. If the minimum value of decoding delay is less than $16$, then 
there exists an LPCOD of size $[2x,9,x]$ with $x\leq 7$ and therefore an LPROD of size $[4x,18,2x]$ exists with $x\leq 7$. By Lemma \ref{rod_nrbl}, there exists a normed real bilinear map of size $[4x,18,2x]$ and hence $4x\geq p_{NBL}(18,2x)\geq 18\circ 2x\geq 18$. Therefore, $x\geq 5$. But for $x=5,6$ and $7$, $18 \circ 2x=26,28$ and $30$ respectively.
In each case, $18\circ 2x>4x$.
\end{proof}
It must be noted that the above argument fails to work when number of antennas is more than $9$. However, it is likely that the proposed rate-$\frac{1}{2}$ scaled-CODs are delay-optimal.

%%%%%%%%%%%%%%%%%%%%%%%%%%%%%%%%
\section{PAPR reduction of rate-$\frac{1}{2}$ scaled-CODs}
\label{sec4}
In this section, we study PAPR properties of the scaled-CODs constructed in this paper. Note that in the construction of $TJC_n$~\cite{TJC}, even though the delay is more, there is no zero entry in the design matrix. On the contrary, in our construction of rate-$\frac{1}{2}$ codes, there are zero entries. To be specific, observe that the first eight columns of rate-$\frac{1}{2}$ code $RH_n, n\geq 9$ given by \eqref{gn} contains as many zero as the number of non-zero entries in it, while there is no zero in the remaining columns of the matrix. When the  number of transmit antennas $n$ is more than $7$, the total number of zeros in the codeword matrix is equal to $8(\nu(n)/2)=4\nu(n)$.  Hence the fraction of zeros in the codeword matrix is equal to $\frac{4\nu(n)}{n\nu(n)}=4/n$ for $n\geq 8$.

Now in the remaining part of this section, we show that one can further reduce the number of zeros in $RH_n$ by suitably choosing a post-multiplication matrix without increasing signaling complexity of the code.

As seen easily, only  the first eight columns contain zeros while the others do not. Moreover, the zeros in the $0$-th column and the $7$-th column occupy complementary locations, so is also for the pairs of columns given by $(1,6), (2,5)$ and $(3,4)$. What it essentially suggests is that we can perform some elementary column operations which will result in a code with no zero entry in it.
%in which all the entries are non-zero.
Let $Q_n$ be an $n\times n$ matrix given by
\begin{eqnarray*}
\label{matrix8}
Q_n=\left[\begin{array}{cc}
A & 0  \\
0 & I_{n-8}
\end{array}\right]
\end{eqnarray*}
where %$A$ is an $8\times 8$ square matrix given by
$I_{n-8}$ is the $(n-8)\times (n-8)$ identity matrix and the matrix $A$ (with entries $0,1$ and $-1$) is given by
\begin{eqnarray*}
A=\frac{1}{\sqrt{2}}\left[\begin{array}{r@{\hspace{0.9pt}}r@{\hspace{0.9pt}}r@{\hspace{0.9pt}}r@{\hspace{0.9pt}}r@{\hspace{0.9pt}}r@{\hspace{0.9pt}}r@{\hspace{0.9pt}}r@{\hspace{0.9pt}}}
  1&0&0&0&0&0&0&1 \\
  0&1&0&0&0&0&1&0  \\
  0&0&1&0&0&1&0&0  \\
  0&0&0&1&1&0&0&0  \\
  0&0&0&1&-&0&0&0 \\
  0&0&1&0&0&-&0&0 \\
  0&1&0&0&0&0&-&0 \\
  1&0&0&0&0&0&0&-
\end{array}\right].
\end{eqnarray*}
Here $-1$ is represented by simply the minus sign.
We post-multiply $RH_n$ with $Q_n$ to get a code in which none of the  entries is 
zero.
%then all the entries of the scaled-COD given by $RH_nQ_n,$ are non-zero. 
We formally present this fact as:
%%%%%%%%%%%%%%%%%%%%%%%%%%%%
\begin{theorem}
$RH_nQ_n$ is a scaled-COD with no zero entry in it. Moreover, the matrix $Q_n$ does not depend on any particular construction procedure (namely the maps $\gamma_t$ and $\psi_t$) used to obtain the constituent rate-1 RODs.
\end{theorem}
%%%%%%%%%%%%%%%%%%%%%%%%%%%%%%%%%%%
\begin{proof}
It is clear that the first $8$ columns of the matrix has $50\%$ zeros in it and in the remaining $n-8$ columns formed by $H_t$ and $\hat{H}_t$, there is no zero as both these matrices are constructed from rate-1 ROD by substituting all the variables in it with appropriate $8$-tuple column vectors. Here neither rate-1 ROD nor the $8$-tuple column vector has any any zero in it. Therefore, the matrix $Q_n$ gives a rate-$\frac{1}{2}$ scaled-COD without any zero irrespective of how the rate-1 RODs are obtained for the construction of $RH_n$.
\end{proof}
%%%%%%%%%%%%%%%%%%%
%%%%%%%%%%%%%%%%%%%%%%%%%%%%%
\begin{example}
For $9$ antennas, we construct a rate-$\frac{1}{2}$ scaled-COD with no zero entry as shown below  

{\footnotesize
\begin{equation*}
  \left[\begin{array}{rrrrrrrrr}
  x_0 &  -x_1^*&  -x_2^*&  -x_3^*   &x_3^*& -x_2^*& -x_1^* &  x_0   & -x_7^*\\
  x_1   &  x_0^*   & -x_3^*   & -x_2^*   & -x_2^*   &  x_3^*   &  x_0^*   & x_1   &  x_6^*\\
  x_2   & -x_3^*   &  x_0^*   &  x_1^*   &  x_1^*   &  x_0^*   &  x_3^*   & x_2   & -x_5^*\\
 -x_3^*   & x_2   &-x_1   & x_0   & x_0   &-x_1   & x_2   &  x_3^*   &-x_4\\
  x_3   &  x_2^*   &  x_1^*   &  x_0^*   & -x_0^*   & -x_1^*   & -x_2^*   & x_3   &  x_4^*\\
  x_2^*   & x_3   & x_0   &-x_1   & x_1   &-x_0   & x_3   & -x_2^*   &-x_5\\
 -x_1^*   & x_0   & x_3   &-x_2   & x_2   & x_3   &-x_0   &  x_1^*   &-x_6\\
  x_0^*   & x_1   &-x_2   & x_3   & x_3   & x_2   &-x_1   & -x_0^*   &-x_7\\
  x_4   & -x_5^*   & -x_6^*   & -x_7^*   & -x_7^*   & -x_6^*   & -x_5^*   & x_4   & -x_3^*\\
  x_5   &  x_4^*   & -x_7^*   & -x_6^*   &  x_6^*   &  x_7^*   &  x_4^*   & x_5   &  x_2^*\\
  x_6   & -x_7^*   &  x_4^*   &  x_5^*   & -x_5^*   &  x_4^*   &  x_7^*   & x_6   & -x_1^*\\
 -x_7^*   & x_6   &-x_5   & x_4   &-x_4   &-x_5   & x_6   &  x_7^*   &-x_0\\
  x_7   &  x_6^*   &  x_5^*   &  x_4^*   &  x_4^*   & -x_5^*   & -x_6^*   & x_7   &  x_0^*\\
  x_6^*   & x_7   & x_4   &-x_5   &-x_5   &-x_4   & x_7   & -x_6^*   &-x_1\\
 -x_5^*   & x_4   & x_7   &-x_6   &-x_6   & x_7   &-x_4   &  x_5^*   &-x_2\\
  x_4^*   & x_5   &-x_6   & x_7   &-x_7   & x_6   &-x_5   & -x_4^*   &-x_3
\end{array}\right]
\end{equation*}
}

\noindent
with each entry  multiplied by $\sqrt{2},$  by post-multiplying the matrix $RH_9$ (given by
the L.H.S of \eqref{ratehalf_9}) with $Q_9$.
\end{example}
%%%%%%%%%%%%%%%%%%%%%%%%%%%%%%%%%%%%%%%%%%%%%%%%%%%%%%%%%%%%%%%%%
\section{Discussion}
\label{sec5}

For any positive integer $n$, this paper gives a rate-$\frac{1}{2}$ scaled-COD for $n$ transmit antennas with decoding delay $\nu(n)$. The decoding delay of these codes is half the decoding delay of the rate-$\frac{1}{2}$ scaled-CODs given by Tarokh et al~\cite{TJC}.
%As the maximal rate of CODs is close to $1/2$ 
When number of transmit antennas is large, the maximal rate of CODs is close to $1/2$ and therefore the rate-$\frac{1}{2}$ codes and the maximal-rate CODs are comparable with respect to the rate of the codes. However, the proposed rate-$\frac{1}{2}$ codes have much less decoding delay than that of the maximal-rate CODs. 
Another advantage with the designs reported in this paper is that they do not contain zero entry leading to low PAPR.

All the four constructions namely Adams, Lax and Phillips's construction from Quaternions \& Octonion, Geramita-Pullman construction and the construction given in this paper will give the same square ROD if number of transmit antennas is less than or equal to $8$. Therefore, these four constructions will generate the same rate-$\frac{1}{2}$ scaled-COD if the number of transmit antennas (of the scaled-COD) is less than or equal to $16$. For more than $16$ antennas, rate-$\frac{1}{2}$ scaled-CODs will vary with the methods chosen for the  construction of rate-1 RODs.
Due to space constraint, 
%Though it is possible to display
two distinct rate-$\frac{1}{2}$ scaled-CODs for $17$ transmit antennas obtained by two different construction procedures for rate-1 RODs, are not given in this paper.
% 
% these designs are not given as it would take up too much space.

It is not known whether the decoding delay of the proposed rate-$\frac{1}{2}$ scaled-COD for given number of transmit antennas is of minimal delay. It is shown that the proposed code for $9$ antennas is of minimal delay.
In general,
we conjecture that $\nu(n)$ is the minimum value of the decoding delay of rate-$\frac{1}{2}$ scaled-COD for any $n$ transmit antennas. It will be interesting to see whether this is indeed true.

%%%%%%%%%%%%%%%%%%%%%%%%%%%%%%%%%%%%%%%%%%%%%%%%%%%%%%

%%%%%%%%%%%%%%%%%%%%%%%%%%%%%%%%%%%%%%%%%%%%%% 
%%%%%%%%%%%%%%%%%%%%%%%%%%%%%%%%%%%

\begin{appendices}
%%%%%%%%%%%%%%%%%%%%%%%%%%%%%%%%%%%%%%%%%%
\section{Recursive Construction of $R_t$}
\label{appendixI}
In this appendix we show that the RODs $R_t$ can be constructed recursively. 

Let $K_t=B_t$ for $t=1,2,4$ and $8$.
The four square ODs $K_t,t=1,2,4,8$ are shown below.

{\small
\begin{eqnarray}
\label{K1248}
(x_0),~
\left(\begin{array}{rr}
    x_0   & x_1\\
   -x_1   & x_0
\end{array}\right), ~
\left(\begin{array}{rrrr}
    x_0   & x_1   & x_2    & x_3 \\
   -x_1   & x_0   &-x_3    & x_2 \\
   -x_2   & x_3   & x_0    &-x_1 \\
   -x_3   &-x_2   & x_1    & x_0 \nonumber\\
\end{array}\right),\\
\left(\begin{array}{rrrrrrrr}
    x_0   & x_1   & x_2    & x_3    & x_4    & x_5     & x_6  & x_7\\
   -x_1   & x_0   &-x_3    & x_2    &-x_5    & x_4     & x_7  &-x_6 \\
   -x_2   & x_3   & x_0    &-x_1    &-x_6    &-x_7     & x_4  & x_5 \\
   -x_3   &-x_2   & x_1    & x_0    &-x_7    & x_6     &-x_5  & x_4 \\
   -x_4   & x_5   & x_6    & x_7    & x_0    &-x_1     &-x_2  &-x_3 \\
   -x_5   &-x_4   & x_7    &-x_6    & x_1    & x_0     & x_3  &-x_2 \\
   -x_6   &-x_7   &-x_4    & x_5    & x_2    &-x_3     & x_0  & x_1 \\
   -x_7   & x_6   &-x_5    &-x_4    & x_3    & x_2     &-x_1  & x_0
\end{array}\right).
\end{eqnarray}
}

It follows that
\begin{eqnarray*}
\begin{array}{c}
K_t^\mathcal{T}=K_t^\mathcal{T}(x_0,x_1,\cdots,x_{t-1})=K_t(x_0,-x_1,\cdots,-x_{t-1})\\\mbox{ and }-K_t^\mathcal{T}=K_t(-x_0,x_1,\cdots,x_{t-1})
\end{array}
\end{eqnarray*}
for $t=1,2,4$ or $8$.
The expression for $R_t$  of order $t$ as given in Theorem \ref{orthog} gives rise to the following recursive construction of $R_t$.
%%%%%%%%%%%%%%%%%%%%%
\begin{figure*}
\begin{equation*}
  R_{2n}= \left[\begin{array}{cc}
 R_n & x_{\rho(n)}I_n \\
-x_{\rho(n)}I_n & R_n^\mathcal{T}
\end{array}\right],
\quad
  R_{4n}= \left[\begin{array}{cc}
 R_{2n} & x_{\rho(n)+1}I_{2n}\\
-x_{\rho(n)+1}I_{2n} & R_{2n}^\mathcal{T}
\end{array}\right],
  \end{equation*}
 \begin{equation}
\label{Rn16n1}
  R_{8n}= \left[\begin{array}{cc}
 R_{4n} & T_4(y_0,y_1)\otimes I_n \\
T_4(-y_0,y_1)\otimes I_n & R_{4n}^\mathcal{T}
\end{array}\right],
\quad
 R_{16n}= \left[\begin{array}{cc}
 R_{8n} & T_8(y_2,y_3,y_4,y_5)\otimes I_n \\
 T_8(-y_2,y_3,y_4,y_5)\otimes I_n & R_{8n}^\mathcal{T}
\end{array}\right]\\
\end{equation}
\hrule
\end{figure*}
%%%%%%%%%%%%%%%%%
Given two matrices $U=(u_{ij})$ of size $v_1\times w_1$ and $V$ of size $v_2\times w_2$, we define the {\it{Kronecker product or tensor product }} of $U$ and $V$ as the following $v_1v_2\times w_1w_2$ matrix:
\[\left(\begin{array}{cccc}
  u_{11}V     &u_{12}V         & \cdots &u_{1w_1}V  \\
  u_{11}V     &u_{12}V         & \cdots &u_{1w_1}V  \\
  \vdots      &\vdots          & \ddots &\vdots     \\
  u_{v_11}V   &u_{v_12}V       & \cdots &u_{v_1w_1}V
\end{array}\right).  \]

Let $I_n$ be an identity matrix of size $n$. Define
\begin{eqnarray*}
I_2^0=\begin{bmatrix} 1 & 0 \\ 0 & 1 \end{bmatrix},&&~
I_2^1=\begin{bmatrix} 1 & 0 \\ 0 & -1 \end{bmatrix},\\
I_2^2=\begin{bmatrix} 0 & 1 \\ 1 & 0 \end{bmatrix},&& ~
I_2^3=\begin{bmatrix} 0 & -1 \\ 1 & 0 \end{bmatrix},\\
I_4^0=I_4,&& I_4^1=I_2^3 \otimes I_2^2,\\
I_8^0=I_8,&& I_8^1=I_2^0 \otimes I_4^1,\\
I_8^2= I_2^3 \otimes I_2^1 \otimes I_2^2,&& I_8^3= I_2^3 \otimes I_2^2 \otimes I_2^0.
\end{eqnarray*}

Let $y_0,\cdots,y_5$ be real variables. Define
\begin{eqnarray*}
T_4(y_0,y_1)&=&y_0I_4^0 +y_1I_4^1,\\
T_8(y_2,y_3,y_4,y_5)&=& y_2 I_8^0 + y_3 I_8^1 +y_4 I_8^2 + y_5 I_8^3.
\end{eqnarray*}

We have four RODs of order $n=2^a$ with $a=0,1,2,3$ as given by
\eqref{K1248} which are respectively $K_1,K_2,K_4$ and $K_8$.\\
Assuming that a square ROD of order $n=2^{4l-1}, l\geq 1$
\[
R_n=R_n(x_0,\cdots,x_{\rho(n)-1})
\]
which has $\rho(n)$ real variables, is given, then
we construct $R_{2n},R_{4n},R_{8n},R_{16n}$ of order $2n$, $4n$, $8n$ and $16n$ respectively given by \eqref{Rn16n1}
%, as shown at the top of the next page
where $y_i=x_{\rho(n)+2+i}$ and
\begin{eqnarray*}
R_t^\mathcal{T}&=&R_t^\mathcal{T}(x_0,x_1,\cdots,x_{\rho(t)-1})\\
               &=&R_t(x_0,-x_1,\cdots,-x_{\rho(t)-1}),\\
-R_t^\mathcal{T}&=&R_t(-x_0,x_1,\cdots,x_{\rho(t)-1}).
\end{eqnarray*}
 
%%%%%%%%%%%%%%%%%%%%%%%%%%%%%%%%%%%%
\section{Adams-Lax-Phillips and Geramita-Pullman constructions as special cases}
%\label{ALPandGP}
\label{appendixII}
In this appendix we show that the well-known constructions of square RODs by Adams-Lax-Phillips using Octonions and Quaternions as well as the construction by Geramita and Pullman are nothing but our construction corresponding to specific choices of the functions  $\gamma_t$ and $\psi_t$ defined by \eqref{gammamap} and \eqref{psimap}.
It turns out to be convenient to use the map  $\chi_t=\psi_t\gamma_t$ than the map $\psi_t.$ Note that both $\gamma_t$ and $\chi_t$ act on the set $Z_{\rho(t)}$ and are injective. Now given $\gamma_t$ and $\chi_t$, we have $\psi_t=\chi_t\gamma_t^{(-1)}$.
With this new definition, we can reformulate the criterion given in Theorem \ref{orthog}
as follows.
\begin{eqnarray}
\label{gammachi}
\vert (\chi_t(x)\oplus \chi_t(y))\cdot(\gamma_t(x)\oplus \gamma_t(y))\vert \\
\mbox{ is an odd integer } \forall x,y\in Z_{\rho(t)}, x\neq y \nonumber.
\end{eqnarray}
In the following lemma, we define $\gamma_t$ and $\chi_t$ in three different ways and these maps are shown to satisfy the relation given by \eqref{gammachi}.
Although both $\gamma_t$ and $\chi_t$ are different for all the three cases for arbitrary values of $t$, $\gamma_t$ is the identity map when $t=1,2,4$ or $8$.
Hence $\chi_t=\psi_t$ if $t\in\{1,2,4,8\}$. % and is given by \eqref{psi1248}.

\begin{lemma}
\label{threedesign}
Let $t=2^a$, $a=4c+d$, $m\in\{0,1,\cdots,7\}$.
Let $\gamma_t$ and $\chi_t$ be two maps defined over $Z_{\rho(t)}$ in three different ways as given below. Identify $\gamma_t(Z_{\rho(t)})$ and $\chi_t(Z_{\rho(t)})$ as subsets of $\mathbb{F}_2^a$.
Then $\vert (\gamma_t(x_1)\oplus \gamma_t(x_2))\cdot (\chi_t(x_1)\oplus \chi_t(x_2))\vert$ is odd for all $x_1,x_2\in Z_{\rho(t)}$, $x_1\neq x_2$. For $x=8l+m \in Z_{\rho(t)},$\\
%%%%%%%%%%%%%%%%%%%%%%%%%
(i)
\begin{eqnarray*}
\begin{array}{lll}
\gamma_t(8l+m)&=&t(1-2^{-l})+{8^l}m\\
\chi_t(8l+m)&=&
\begin{cases}
0  &  \text{ if  $l=0, m=0$} \\
t.2^{-l} & \text{ if  $l\neq 0, m=0$} \\
8^l\chi_{2^d}(m)&\text{ if $l=c, m\neq 0$ }\\
t.2^{-l-1}+8^l\chi_8(m)&\text{ if $l\neq c, m\neq 0$ }
 \end{cases}
\end{array}
\end{eqnarray*}
%%%%%%%%%%%%%%%%%%%%%%%%%%%%
(ii)
{\footnotesize
\begin{eqnarray*}
\gamma_t(8l+m)&=&
 \begin{cases}
t(1-2^{-2l})+2^{2l}m  & \hspace{-8pt}\mbox{ if }  0\leq m \leq 3  \\
t(1-2^{-2l-1})+2^{2l}(m-4)  & \hspace{-8pt}\mbox{ if } 4\leq m \leq 7,  \nonumber\\
\end{cases}\\
\chi_t(8l+m)&=&
\begin{cases}
0  &  \mbox{ if } l=0,m=0 \\
t.2^{-2l} & \text{ if  $l\neq 0,m=0$}\\
t.2^{-2l-1}  &  \text{ if  $l\neq 0,m=4$} \\
4            &  \text{ if  $l=0,m=4$} \\
2^{2l}\chi_{2^d}(m)&\text{ if $l=c,m\neq 0$ }\\
t.2^{-2l-1}+2^{2l}\chi_4(m)&\text{ if $l\neq c,m\in\{1,2,3\}$ }\\
t.2^{-2l-2}+2^{2l}\chi_4^\prime(m-4)&\text{ if $l\neq c,m\in\{5,6,7\}$ }\\
 \end{cases}
\end{eqnarray*}
}
where $\chi_4^\prime=\left(\begin{array}{cccccccccccc}
    0   & 1 &2 &3 \\
    0   & 1 &3 &2
    \end{array}\right),$
%%%%%%%%%%%%%

(iii)
\begin{eqnarray*}
\gamma_t(8l+m)&=&
 \begin{cases}
\frac{8t}{15}(1-2^{-4l}) +\frac{tm}{16^{l+1}} &\text{ if $l< c$, }\\
\frac{8t}{15}(1-2^{-4l}) +m & \text{ if $l=c$ }\nonumber\\
\end{cases}\\
\chi_t(8l+m)&=&
\begin{cases}
0  &  \text{ if  $l=0,m=0$} \\
\frac{t}{2}2^{-4(l-1)} & \text{ if  $l\neq 0, m=0$}\\
 \chi_{2^d}(m)&\text{ if $l=c, m\neq 0$. }\\
\frac{t}{2}2^{-4l}+\frac{t\chi_8(m)}{2^{4(l+1)}}&\text{ if $l\neq c,m\neq 0$. }
 \end{cases}
\end{eqnarray*}

\end{lemma}
%%%%%%%%%%%%%%%%%%%%%%%%%%%%%%%%%%%%%%%%%%%%%%%%%%%
\begin{proof}
We give proof only for the case (i). The cases  (ii) and (iii) can be proved similarly.\\
It is enough to prove that \\
  (B1) $\vert \gamma_t(x)\cdot \chi_t(x)\vert$ is odd for all $x\neq 0$, $x\in Z_{\rho(t)}$ and \\
  (B2) $\vert \gamma_t(x_1)\cdot \chi_t(x_2)\vert +\vert\gamma_t(x_2)\cdot \chi_t(x_1)\vert$ is odd for all $x_1,x_2\in Z_{\rho(t)}$,  $x_1\neq x_2, x_1\neq 0,x_2\neq 0$.

Let $\gamma_t(8l+m)=\gamma_t^{(1)}(8l+m)+\gamma_t^{(2)}(8l+m)$ such that $\gamma_t^{(1)}(8l+m)=t(1-2^{-l})$ and $\gamma_t^{(2)}(8l+m)={8^l}m$.\\
Similarly, let $\chi_t(8l+m)=\chi_t^{(1)}(8l+m)+\chi_t^{(2)}(8l+m)$ such that 
\begin{eqnarray}
\chi_t^{(1)}(8l+m)=
    \begin{cases}
0  &  \text{ if  $l=0,m=0,$} \nonumber\\
t2^{-l} & \text{ if  $l\neq 0,m=0,$} \nonumber\\
0 &\text{ if $l=c,m \neq 0,$ }\\
t2^{-l-1} &\text{ if $l\neq c,m\neq 0,$ }\\
 \end{cases}\\
\chi_t^{(2)}(8l+m)=
    \begin{cases}
0  &  \text{ if  $l=0,m=0,$} \nonumber\\
0  & \text{ if  $l\neq 0,m=0,$} \nonumber\\
8^l\chi_{2^d}(m)&\text{ if $l=c,m\neq 0,$ }\\
8^l\chi_{8}(m)&\text{ if $l\neq c,m\neq 0.$ }
 \end{cases}
\end{eqnarray}

\noindent Let $8l+m\neq 0$ and $8l^\prime+ m^\prime\neq 0$.
From the definition of $\gamma_t^i,\chi_t^i,i=1,2$, it follows that

{\small
\begin{eqnarray*}
\begin{array}{l}
(A1)~~ \vert\chi_t^{(2)}(8l+m)\cdot \gamma_t^{(2)}(8l^\prime+ m^\prime)\vert=0 \mbox{ if } l\neq l^\prime,\\
(A2)~~ \vert\chi_t^{(1)}(8l+m)\cdot \gamma_t^{(1)}(8l^\prime+ m^\prime)\vert=1\mbox{ if }l <l^\prime,\\
(A3)~~ \vert\chi_t^{(1)}(8l+m)\cdot \gamma_t^{(1)}(8l^\prime+ m^\prime)\vert=0\mbox{ if }l >l^\prime \\
\hspace{24pt}\mbox{ or if } l=l^\prime, m\neq 0, \\% provided both $m$ and  $m^\prime$
(A4)~~ \vert\chi_t^{(1)}(8l)\cdot \gamma_t^{(1)}(8l+m)\vert=1 \mbox{ if }l\neq 0, \\
(A5)~~ \vert\chi_t^{(1)}(x)\cdot \gamma_t^{(2)}(y)\vert=\vert\chi_t^{(2)}(x)\cdot \gamma_t^{(1)}(y)\vert=0
     ~\forall~ x,y\in Z_{\rho(t)},\\
(A6)~~ \vert\chi_t^{(2)}(8l)\cdot \gamma_t^{(2)}(8l+m)\vert=\vert\chi_t^{(2)}(8l+m)\cdot \gamma_t^{(2)}(8l)\vert=0.
\end{array}
\end{eqnarray*}
}

First we prove (B1). Let $x=8l+m$ with $m\neq 0$. We have 

{\small
\begin{eqnarray*}
\vert \chi_t(x)\cdot\gamma_t(x)\vert&\equiv&\vert\chi_t^{(1)}(8l+m)\cdot \gamma_t^{(1)}(8l+m)\vert + \vert\chi_t^{(2)}(8l+m)\\
&& \cdot \gamma_t^{(2)}(8l+m)\vert + \vert\chi_t^{(1)}(8l+m)\cdot \gamma_t^{(2)}(8l+m)\vert \\
&& + \vert\chi_t^{(2)}(8l+m)\cdot \gamma_t^{(1)}(8l+m)\vert\\
&=&\vert\chi_t^{(1)}(8l+m)\cdot \gamma_t^{(1)}(8l+m)\vert\\ &&+\vert\chi_t^{(2)}(8l+m)\cdot \gamma_t^{(2)}(8l+m)\vert \mbox{ by (A5) } \\
 &=&\vert\chi_t^{(2)}(8l+m)\cdot \gamma_t^{(2)}(8l+m)\vert \mbox{ using (A3) }\\	
&=&\vert\chi_e(m)\cdot m\vert, \mbox{ $e=2^d$ if $l=c$, else $e=8$ }
\end{eqnarray*}
}
\noindent
But $\vert\chi_e(m)\cdot m\vert$ is an odd number by Lemma \ref{propertyV8}.\\
If $m=0$, we have $\vert \gamma_t(x)\cdot \chi_t(x)\vert=1$ by (A4).\\

To prove (B2), let $x_1\neq 0$ and $x_2\neq 0$.
Write $x_2=8l_2+m_2$, $x_1=8l_1+m_1$ with $x_2>x_1$. %, x_2\neq 0 and x_1\neq 0$. 
We have two cases:\\
(C1): $l_2>l_1$,~~~~~~
(C2): $l_2=l_1=l$, $m_2> m_1$.\\

\noindent
{\bf Case (C1):} we have 

{\small
\begin{eqnarray*}
\chi_t(x_2)\cdot \gamma_t(x_1)=\chi_t^{(1)}(8l_2+m_2)\cdot \gamma_t^{(1)}(8l_1+m_1) \\
\vspace{120pt}\oplus \chi_t^{(2)}(8l_2+m_2)\cdot \gamma_t^{(2)}(8l_1+m_1) \text{ by (A5) }.
\end{eqnarray*}
}
\noindent
But $\vert\chi_t^{(1)}(8l_2+m_2)\cdot \gamma_t^{(1)}(8l_1+m_1)\vert=0$ by (A3) \\
and $\vert\chi_t^{(2)}(8l_2+m_2)\cdot \gamma_t^{(2)}(8l_1+m_1)\vert=0$ by (A1),\\
thus $\vert\chi_t(x_2)\cdot \gamma_t(x_1)\vert=0$. \\
Now $\chi_t(x_1)\cdot \gamma_t(x_2)= \chi_t^{(1)}(8l_1+m_1)\cdot \gamma_t^{(1)}(8l_2+m_2) \oplus \chi_t^{(2)}(8l_1+m_1)\cdot \gamma_t^{(2)}(8l_2+m_2)$ by (A5).\\
But $\vert\chi_t^{(2)}(8l_1+m_1)\cdot \gamma_t^{(2)}(8l_2+m_2)\vert=0$ by (A1) and
$\vert \chi_t^{(1)}(8l_1+m_1)\cdot \gamma_t^{(1)}(8l_2+m_2)\vert=1$ by (A2).\\
Hence $\vert \chi_t(x_1)\cdot \gamma_t(x_2)\vert + \vert \chi_t(x_2)\cdot \gamma_t(x_1)\vert$
is an odd number.

\noindent
{\textbf{Case (C2):}} we consider two following cases:\\
 (i) $m_1\neq 0$ and (ii) $m_1=0$. Note that $m_2$ is always non-zero.\\
Let $d=\vert (\chi_t(x_1)\cdot \gamma_t(x_2))\oplus(\chi_t(x_2)\cdot \gamma_t(x_1))\vert$.\\

\noindent
{\textbf{Case (i):}} We have
 {\footnotesize
\begin{eqnarray*}
 d &\equiv&\vert\chi_t^{(2)}(8l+m_1)\cdot\gamma_t^{(2)}(8l+m_2)\vert \\
&&+ \vert \chi_t^{(2)}(8l+m_2)\cdot \gamma_t^{(2)}(8l+m_1)\vert \mbox{ by (A3) and (A5) }\\
&=&\vert(\chi_e(m_1)\cdot m_2)
\oplus(\chi_e(m_2)\cdot m_1)\vert, \mbox{ $e=2^d$ if $l=c$, else $e=8$ }
\end{eqnarray*}
}
which is an odd number by Lemma \ref{propertyV8}.\\
{\textbf{Case (ii):}}
Since $m_1=0$, therefore $l\neq 0$. We have
\begin{eqnarray*}
 d &\equiv&\vert\chi_t^{(1)}(8l)\cdot\gamma_t^{(2)}(8l+m_2)\vert \\
&&+\vert \chi_t^{(1)}(8l+m_2)\cdot \gamma_t^{(1)}(8l)\vert \mbox{ by (A6). }\\
&=& 1 \mbox{ by (A3) and (A4). }	
\end{eqnarray*}
\end{proof}

%%%%%%%%%%%%%%%%%%%%%%%%%%%%%%%%%%%%%%%%%%%%%%%%%%%%%%%%%%%%
\begin{figure*}
{\small
\begin{equation}
\label{ALPQ}
  %\mathbb{O}_3=
\mathbb{O}_{16n}^{(Q)}=\left(\begin{array}{cccc}
I_n \otimes L_4(x_0,x_1,x_2,x_3)& 0_{4n} & I_n \otimes R_4(x_4,x_5,x_6,x_7) & \mathbb{O}_1(y_0,\cdots,y_{\rho(n)-1})\otimes I_4 \\
  0_{4n} & I_n \otimes L_4(x_0,x_1,x_2,x_3) & -\mathbb{O}_1^\mathcal{T}(y_0,\cdots,y_{\rho(n)-1})\otimes I_4  & I_n \otimes R_4^\mathcal{T}(x_4,x_5,x_6,x_7) \\
 I_n \otimes -R_4^\mathcal{T}(x_4,x_5,x_6,x_7) & \mathbb{O}_1(y_0,\cdots,y_{\rho(n)-1})\otimes I_4 & I_n \otimes L_4^\mathcal{T}(x_0,x_1,x_2,x_3)
& 0_{4n} \\
-\mathbb{O}_1^\mathcal{T}(y_0,\cdots,y_{\rho(n)-1})\otimes I_4 & I_n \otimes -R_4(x_4,x_5,x_6,x_7) &  0_{4n} & I_n\otimes L_4^\mathcal{T}(x_0,x_1,x_2,x_3)
\end{array}\right)
\end{equation}
}
\hrule
\end{figure*}
%%%%%%%%%%%%%%%%%%%%%%%%%%%%%%%%%
By Lemma \ref{threedesign} and Theorem \ref{ratesquarerod}, the matrix $B_t$ defined by two functions $\gamma_t$ and $\chi_t$ is a square ROD in all the three cases. We refer to these three different RODs by $A_t, \hat{A}_t$ and $P_t$ corresponding to the pair of functions defined in (i), (ii) and (iii) respectively. 
% Observe that our construction $R_t$ is different from any of $A_t, \hat{A}_t$ and $P_t$ for general values of $t.$

Now, we proceed to show that the designs $A_t, \hat{A}_t$ and $P_t$ are essentially the Adams-Lax-Phillips construction using Octonions and Quaternions and the Geramita-Pullman construction respectively with change in sign of some rows or columns. 
 
\subsection{Adams-Lax-Phillips Construction from Octonions as a special case}
The Adams-Lax-Phillips construction from Octonions is given by induction from order $n=2^a$ to $16n$ as follows~\cite{Lia}:
denoting the square ROD of order $n=2^a$ resulting from the Adams-Lax-Phillips construction using Octonions by  
\[
\mathbb{O}_n=\mathbb{O}_n(x_0,\cdots,x_{\rho(n)-1})
\]
which has $\rho(n)$ real variables, the square ROD of order $16n$ with $(\rho(n)+8)$
real variables $x_i,$ $i=0,1,\cdots, \rho(n)+7,$
\[
\mathbb{O}_{16n}=\mathbb{O}_{16n}(x_0,\cdots,x_{\rho(n)+7})
\]
is given by
\begin{eqnarray*}
  \mathbb{O}_{16n}= \left[\begin{array}{cc}
 I_n\otimes K_8(y_0,\cdots,y_7) & \mathbb{O}_n\otimes I_8 \\
\mathbb{O}_n^\mathcal{T}\otimes I_8 & I_n\otimes (-K_8^\mathcal{T}(y_0,\cdots,y_7))\\
\end{array}\right]
 \end{eqnarray*}
with $y_i=x_{\rho(n)+i}$. \\
With re-arrangement of variables and change in signs, we rewrite the design $\mathbb{O}_{16n}$ as

{\footnotesize
\begin{eqnarray*}
%\label{ALPO}
  \mathbb{O}_{16n}^{(O)}= \left[\begin{array}%{cc}
{c@{\hspace{0.4pt}}c@{\hspace{0.2pt}}}
 I_n\otimes K_8(x_0,\cdots,x_7) & \mathbb{O}_n^{(O)}(y_0,\cdots,y_{\rho(n)-1})\otimes I_8 \\
-\mathbb{O}_n^{(O)\mathcal{T}}(y_0,\cdots,y_{\rho(n)-1})\otimes I_8 & I_n\otimes K_8^\mathcal{T}(x_0,\cdots,x_7)\\
\end{array}\right]
 \end{eqnarray*}
}
with $y_i=x_{8+i}$ and $\mathbb{O}_n^{(O)}=\mathbb{O}_n,~~ n=1,2,4,8.$ The reason why we consider this rearranged version is that we show in Lemma \ref{AtO2} that $A_t$ is same as $\mathbb{O}_{2n}^{(O)}$ with $t=16n.$\\
%%%%%%%%%%%%%%%%%%%%%%%%
\begin{lemma}
\label{AtO2}
Let $t \geq 16$ be a power of $2.$ Also, let $A_t$ be the square ROD of order $t$ as given in Lemma \ref{threedesign} (i), and $\mathbb{O}_{16n}^{(O)}$ be the square ROD 
%given by \eqref{ALPO} 
which is of order $16n$. Then $A_t=\mathbb{O}_{16n}^{(O)}$ for $t=16n$.
\end{lemma}
%%%%%%%%%%%%%%%%%%% 
\begin{proof} We prove it by induction on $t$.
For $t=1,2,4$ and $8$, $A_t=K_t$ and the COD $\mathbb{O}_t^{(O)}$ of order $t$ is also given by $K_t$. Hence the lemma holds for $t=1,2,4$ and $8$.
Assuming that the lemma holds for $t=n$, i.e., $A_n=\mathbb{O}_n^{(O)}$ of order $n$, we have to prove that the lemma also holds for $t=16n, $ i.e., $A_{16n}=\mathbb{O}_{16n}^{(O)}$. Let
\begin{eqnarray*}
  A_{16n}= \left[\begin{array}{cc}
 \hat{A}_{11} & \hat{A}_{12} \\
\hat{A}_{21} & \hat{A}_{22}
\end{array}\right]
 \end{eqnarray*}
where $\hat{A}_{ij}$, $1\leq i,j\leq 2$ are square matrices of size $8n\times 8n$.
It is easy to check that the location of non-zero variables in the matrix $A_{16n}$ coincide with that of $\mathbb{O}_{16n}^{(O)}$. Therefore it is enough to show the signs (positive/negative polarity) of the corresponding entry in the two designs are same i.e.,
\begin{enumerate}
\item $\mu_{16n}(i,j)= \mu_{16n}(i\% 8, j\% 8 )$ for $0\leq i,j \leq 8n-1$,
\item $\mu_{16n}(i,j)= \mu_8(i,j)$ for $0\leq i,j \leq 7$,
\item $\mu_{16n}(i,j)=\mu_{16n}(i\oplus i\%8, j\oplus j\%8)$ \\
if $0\leq i\leq 8n-1, 8n\leq j \leq 16n-1$,
\item $\mu_{16n}(8i,8n\oplus 8j)=\mu_n(i,j)$ for $0\leq i,j \leq n-1$, 
\item $\mu_{16n}(8n \oplus i,8n \oplus j)=\mu_{16n}(i,j)$ if $i\oplus j=0$ or $i\oplus j>8n$,
\item $\mu_{16n}(8n \oplus i,8n \oplus j)=-\mu_{16n}(i,j)$ 
if $i\oplus j\in\{1,2,\cdots,7\}\cup \{8n\}$. % or $i+j>n/2$
\end{enumerate}

Note that\\
1) \& 2) together imply $\hat{A}_{11}=I_n\otimes K_8(x_0,\cdots,x_7 )$,\\
3) \& 4) together imply $\hat{A}_{12}=\mathbb{O}_{n}^{(O)}\otimes I_8$ and \\
5) \& 6) together imply $\hat{A}_{22}=A_{11}^\mathcal{T}, \hat{A}_{21}=-A_{12}^\mathcal{T}$.\\
\noindent
Let $A_{16n}(i,j)\neq 0$.\\
Then $i\oplus j\in \hat{Z}_{\rho(16n)}$ and $\mu_{16n}(i,j)=(-1)^{\vert i\cdot\psi_{16n}(i\oplus j)\vert}$.\\
To prove 1), we have to show that
$\vert i\cdot\psi_{16n}(i\oplus j)\vert\equiv \vert (i\% 8)\cdot\psi_{16n}(i\% 8\oplus j\% 8)$
for $0\leq i,j \leq 8n-1$.\\
We have $i\oplus j=(16n)(1-2^{-l})+{8^l}m$ and $i\oplus j< 8n$. So $l=0$ and $i\oplus j=m$. i.e., $i \oplus j=i\% 8 \oplus j\%8$.\\
Thus it is enough to prove that $\vert (i\oplus i\% 8)\cdot\psi_{16n}(i\oplus j)\vert\equiv 0$
Now
$(i\oplus i\% 8)< 8n$, $8$ divides $(i\oplus i\% 8)$ and $\psi_{16n}(i\oplus j)=8n\oplus\psi_8(m)$, hence the statement holds.

The statement 2) is true as
$\vert i\cdot\psi_{16n}(i\oplus j)\vert\equiv \vert i\cdot\psi_8(i\oplus j)\vert$
for $0\leq i,j \leq 7$. \\ % because
In order to prove 3), we must have
%{ %\small
\begin{eqnarray*}
 \vert i\cdot\psi_{16n}(i\oplus j)\vert \equiv \vert (i\oplus i\% 8)\cdot\psi_{16n}((i \oplus i\% 8) \oplus (j\oplus j\% 8))\vert
\end{eqnarray*}
%}
i.e., $\vert (i\% 8)\cdot\psi_{16n}((i \oplus i\% 8) \oplus (j\oplus j\% 8))\vert\equiv 0$.
As $8n\leq i\oplus j\leq 16n-1$, we have $i\oplus j=(16n)(1-2^{-l})+{8^l}m$ with $l\geq 1$.
So $8$ divides $i \oplus j$ as $8$ divides both $(16n)(1-2^{-l})$ and ${8^l}m$.
So $i\% 8=j\% 8$ i.e., $i\oplus j=((i \oplus i\%8) \oplus (j\oplus j\%8))$.
Thus it is enough to prove that $\vert (i\% 8)\cdot\psi_{16n}(i \oplus j)\vert \equiv 0$.
It is indeed true as $\psi_{16n}(i \oplus j)$ is a multiple of $8$.

To prove 4), we have to show that
\[
\vert (8i)\cdot\psi_{16n}(8n\oplus 8i\oplus 8j)\vert \equiv \vert (i\cdot\psi_{n}((i \oplus j).
\]
We have $8n\oplus 8i\oplus 8j=(16n)(1-2^{-l})+8^lm$ for some $l$ with $l\geq 1$ and $m\in Z_8$.
Let $16n=2^a$ and $a=4c+d$.\\
If $l=c$, we have $\psi_{16n}(8n\oplus 8i\oplus 8j)=8^l\chi_{2^d}(m)$ and $\psi_{n}(i\oplus j)=8^{l-1}\chi_{2^d}(m)$. One can easily see that the above statement holds.\\
On the other hand, if $l<c$, we have  $\psi_{16n}(8n\oplus 8i\oplus 8j)=(16n)2^{-l-1}+8^l\chi_{8}(m)$ and $\psi_{n}(i\oplus j)=n.2^{-l}+8^{l-1}\chi_{8}(m)$.
In this case too, the statement holds.

To prove 5), we have to show that
\[ \vert (i\oplus 8n)\cdot\psi_{16n}(i\oplus j)\vert \equiv \vert i\cdot\psi_{16n}(i\oplus j)\vert,
\]
i.e., $\vert (8n)\cdot\psi_{16n}(i\oplus j)\vert\equiv 0$.
Now for $i\oplus j=0$ or greater than $8n$, $ (8n)\cdot\psi_{16n}(i\oplus j)=\mathbf{0}$.

To prove 6), we have to show that
\[ \vert (i\oplus 8n)\cdot\psi_{16n}(i\oplus j)\vert \equiv 1+\vert i\cdot\psi_{16n}(i\oplus j)\vert,
\]
i.e., $\vert (8n)\cdot\psi_{16n}(i\oplus j)\vert\equiv 1$.
But $(8n)\cdot\psi_{16n}(i\oplus j)=8n$ for all $(i\oplus j)\in\{1,2,3,4,5,6,7,8n\}$.
\end{proof}

\subsection{Adams-Lax-Phillips Construction from Quaternions and Geramita-Pullman 	Construction as special cases}

%%%%%%%%%%
%%%%%%%%%%%%%%%%%%%%%%%%%%%%%%%%%%%
\begin{figure*}
\scriptsize{
\begin{eqnarray}
\label{GP32}
\left[\begin{array}{ r @{\hspace{.2pt}} r @{\hspace{.2pt}} r @{\hspace{.2pt}} r @{\hspace{.2pt}}r @{\hspace{.2pt}} r @{\hspace{.2pt}} r @{\hspace{.2pt}} r @{\hspace{.2pt}}
r @{\hspace{.2pt}} r @{\hspace{.2pt}} r @{\hspace{.2pt}} r @{\hspace{.2pt}}r @{\hspace{.2pt}} r @{\hspace{.2pt}} r @{\hspace{.2pt}} r @{\hspace{.2pt}} r @{\hspace{.2pt}} r @{\hspace{.2pt}} r @{\hspace{.2pt}} r @{\hspace{.2pt}}r @{\hspace{.2pt}} r @{\hspace{.2pt}} r @{\hspace{.2pt}} r @{\hspace{.2pt}}r @{\hspace{.2pt}} r @{\hspace{.2pt}} r @{\hspace{.2pt}} r @{\hspace{.2pt}}r @{\hspace{.2pt}} r @{\hspace{.2pt}} r @{\hspace{.2pt}} r @{\hspace{.2pt}}}
  x_0 &     0 &    x_1 &     0 &    x_2 &     0 &    x_3 &     0 &    x_4 &     0 &    x_5 &     0 &    x_6 &     0 &    x_7 &     0 &    x_8 &    x_9 &     0 &     0 &     0 &     0 &     0 &     0 &     0 &     0 &     0 &     0 &     0 &     0 &     0 &     0\\
   0 &    x_0 &     0 &    x_1 &     0 &    x_2 &     0 &    x_3 &     0 &    x_4 &     0 &    x_5 &     0 &    x_6 &     0 &    x_7 &   -x_9 &    x_8 &     0 &     0 &     0 &     0 &     0 &     0 &     0 &     0 &     0 &     0 &     0 &     0 &     0 &     0\\
 -x_1 &     0 &    x_0 &     0 &   -x_3 &     0 &    x_2 &     0 &   -x_5 &     0 &    x_4 &     0 &    x_7 &     0 &   -x_6 &     0 &     0 &     0 &    x_8 &    x_9 &     0 &     0 &     0 &     0 &     0 &     0 &     0 &     0 &     0 &     0 &     0 &     0\\
   0 &   -x_1 &     0 &    x_0 &     0 &   -x_3 &     0 &    x_2 &     0 &   -x_5 &     0 &    x_4 &     0 &    x_7 &     0 &   -x_6 &     0 &     0 &   -x_9 &    x_8 &     0 &     0 &     0 &     0 &     0 &     0 &     0 &     0 &     0 &     0 &     0 &     0\\
 -x_2 &     0 &    x_3 &     0 &    x_0 &     0 &   -x_1 &     0 &   -x_6 &     0 &   -x_7 &     0 &    x_4 &     0 &    x_5 &     0 &     0 &     0 &     0 &     0 &    x_8 &    x_9 &     0 &     0 &     0 &     0 &     0 &     0 &     0 &     0 &     0 &     0\\
   0 &   -x_2 &     0 &    x_3 &     0 &    x_0 &     0 &   -x_1 &     0 &   -x_6 &     0 &   -x_7 &     0 &    x_4 &     0 &    x_5 &     0 &     0 &     0 &     0 &   -x_9 &    x_8 &     0 &     0 &     0 &     0 &     0 &     0 &     0 &     0 &     0 &     0\\
 -x_3 &     0 &   -x_2 &     0 &    x_1 &     0 &    x_0 &     0 &   -x_7 &     0 &    x_6 &     0 &   -x_5 &     0 &    x_4 &     0 &     0 &     0 &     0 &     0 &     0 &     0 &   x_8 &    x_9 &     0 &     0 &     0 &     0 &     0 &     0 &     0 &     0\\
   0 &   -x_3 &     0 &   -x_2 &     0 &    x_1 &     0 &    x_0 &     0 &   -x_7 &     0 &    x_6 &     0 &   -x_5 &     0 &    x_4 &     0 &     0 &     0 &     0 &     0 &     0 &  -x_9 &    x_8 &     0 &     0 &     0 &     0 &     0 &     0 &     0 &     0\\
 -x_4 &     0 &    x_5 &     0 &    x_6 &     0 &    x_7 &     0 &    x_0 &     0 &   -x_1 &     0 &   -x_2 &     0 &   -x_3 &     0 &     0 &     0 &     0 &     0 &     0 &     0 &    0 &     0 &    x_8 &    x_9 &     0 &     0 &     0 &     0 &     0 &     0\\
   0 &   -x_4 &     0 &    x_5 &     0 &    x_6 &     0 &    x_7 &     0 &    x_0 &     0 &   -x_1 &     0 &   -x_2 &     0 &   -x_3 &     0 &     0 &     0 &     0 &     0 &     0 &    0 &     0 &   -x_9 &    x_8 &     0 &     0 &     0 &     0 &     0 &     0\\
 -x_5 &     0 &   -x_4 &     0 &    x_7 &     0 &   -x_6 &     0 &    x_1 &     0 &    x_0 &     0 &    x_3 &     0 &   -x_2 &     0 &     0 &     0 &     0 &     0 &     0 &     0 &    0 &     0 &     0 &     0 &    x_8 &    x_9 &     0 &     0 &     0 &     0\\
   0 &   -x_5 &     0 &   -x_4 &     0 &    x_7 &     0 &   -x_6 &     0 &    x_1 &     0 &    x_0 &     0 &    x_3 &     0 &   -x_2 &     0 &     0 &     0 &     0 &     0 &     0 &    0 &     0 &     0 &     0 &   -x_9 &    x_8 &     0 &     0 &     0 &     0\\
 -x_6 &     0 &   -x_7 &     0 &   -x_4 &     0 &    x_5 &     0 &    x_2 &     0 &   -x_3 &     0 &    x_0 &     0 &    x_1 &     0 &     0 &     0 &     0 &     0 &     0 &     0 &    0 &     0 &     0 &     0 &     0 &     0 &    x_8 &    x_9 &     0 &     0\\
   0 &   -x_6 &     0 &   -x_7 &     0 &   -x_4 &     0 &    x_5 &     0 &    x_2 &     0 &   -x_3 &     0 &    x_0 &     0 &    x_1 &     0 &     0 &     0 &     0 &     0 &     0 &    0 &     0 &     0 &     0 &     0 &     0 &   -x_9 &    x_8 &     0 &     0\\
 -x_7 &     0 &    x_6 &     0 &   -x_5 &     0 &   -x_4 &     0 &    x_3 &     0 &    x_2 &     0 &   -x_1 &     0 &    x_0 &     0 &     0 &     0 &     0 &     0 &     0 &     0 &    0 &     0 &     0 &     0 &     0 &     0 &     0 &     0 &    x_8 &    x_9\\
   0 &   -x_7 &     0 &    x_6 &     0 &   -x_5 &     0 &   -x_4 &     0 &    x_3 &     0 &    x_2 &     0 &   -x_1 &     0 &    x_0 &     0 &     0 &     0 &     0 &     0 &     0 &    0 &     0 &     0 &     0 &     0 &     0 &     0 &     0 &   -x_9 &    x_8\\
 -x_8 &    x_9 &     0 &     0 &     0 &     0 &     0 &     0 &     0 &     0 &     0 &  0 &     0 &     0 &     0 &     0 &    x_0 &     0 &   -x_1 &     0 &   -x_2 &     0 &   -x_3 &     0 &   -x_4 &     0 &   -x_5 &     0 &   -x_6 &     0 &   -x_7 &     0\\
 -x_9 &   -x_8 &     0 &     0 &     0 &     0 &     0 &     0 &     0 &     0 &     0 &  0 &     0 &     0 &     0 &     0 &     0 &    x_0 &     0 &   -x_1 &     0 &   -x_2 & 0 &   -x_3 &     0 &   -x_4 &     0 &   -x_5 &     0 &   -x_6 &     0 &   -x_7\\
   0 &     0 &   -x_8 &    x_9 &     0 &     0 &     0 &     0 &     0 &     0 &     0 &  0 &     0 &     0 &     0 &     0 &    x_1 &     0 &    x_0 &     0 &    x_3 &     0 &   -x_2 &     0 &    x_5 &     0 &   -x_4 &     0 &   -x_7 &     0 &    x_6 &     0\\
   0 &     0 &   -x_9 &   -x_8 &     0 &     0 &     0 &     0 &     0 &     0 &     0 &  0 &     0 &     0 &     0 &     0 &     0 &    x_1 &     0 &    x_0 &     0 &    x_3 & 0 &   -x_2 &     0 &    x_5 &     0 &   -x_4 &     0 &   -x_7 &     0 &    x_6\\
   0 &     0 &     0 &     0 &   -x_8 &    x_9 &     0 &     0 &     0 &     0 &     0 &  0 &     0 &     0 &     0 &     0 &    x_2 &     0 &   -x_3 &     0 &    x_0 &     0 &    x_1 &     0 &    x_6 &     0 &    x_7 &     0 &   -x_4 &     0 &   -x_5 &     0\\
   0 &     0 &     0 &     0 &   -x_9 &   -x_8 &     0 &     0 &     0 &     0 &     0 &  0 &     0 &     0 &     0 &     0 &     0 &    x_2 &     0 &   -x_3 &     0 &    x_0 & 0 &    x_1 &     0 &    x_6 &     0 &    x_7 &     0 &   -x_4 &     0 &   -x_5\\
   0 &     0 &     0 &     0 &     0 &     0 &   -x_8 &    x_9 &     0 &     0 &     0 &  0 &     0 &     0 &     0 &     0 &    x_3 &     0 &    x_2 &     0 &   -x_1 &     0 &    x_0 &     0 &    x_7 &     0 &   -x_6 &     0 &    x_5 &     0 &   -x_4 &     0\\
   0 &     0 &     0 &     0 &     0 &     0 &   -x_9 &   -x_8 &     0 &     0 &     0 &  0 &     0 &     0 &     0 &     0 &     0 &    x_3 &     0 &    x_2 &     0 &   -x_1 & 0 &    x_0 &     0 &    x_7 &     0 &   -x_6 &     0 &    x_5 &     0 &   -x_4\\
   0 &     0 &     0 &     0 &     0 &     0 &     0 &     0 &   -x_8 &    x_9 &     0 &  0 &     0 &     0 &     0 &     0 &    x_4 &     0 &   -x_5 &     0 &   -x_6 &     0 &   -x_7 &     0 &    x_0 &     0 &    x_1 &     0 &    x_2 &     0 &    x_3 &     0\\
   0 &     0 &     0 &     0 &     0 &     0 &     0 &     0 &   -x_9 &   -x_8 &     0 &  0 &     0 &     0 &     0 &     0 &     0 &    x_4 &     0 &   -x_5 &     0 &   -x_6 & 0 &   -x_7 &     0 &    x_0 &     0 &    x_1 &     0 &    x_2 &     0 &    x_3\\
   0 &     0 &     0 &     0 &     0 &     0 &     0 &     0 &     0 &     0 &   -x_8 &    x_9 &     0 &     0 &     0 &     0 &    x_5 &     0 &    x_4 &     0 &   -x_7 &     0 &    x_6 &     0 &   -x_1 &     0 &    x_0 &     0 &   -x_3 &     0 &    x_2 &     0\\
   0 &     0 &     0 &     0 &     0 &     0 &     0 &     0 &     0 &     0 &   -x_9 &   -x_8 &     0 &     0 &     0 &     0 &     0 &    x_5 &     0 &    x_4 &     0 &   -x_7 & 0 &    x_6 &     0 &   -x_1 &     0 &    x_0 &     0 &   -x_3 &     0 &    x_2\\
   0 &     0 &     0 &     0 &     0 &     0 &     0 &     0 &     0 &     0 &     0 &     0 &   -x_8 &    x_9 &     0 &     0 &    x_6 &     0 &    x_7 &     0 &    x_4 &     0 &   -x_5 &     0 &   -x_2 &     0 &    x_3 &     0 &    x_0 &     0 &   -x_1 &     0\\
   0 &     0 &     0 &     0 &     0 &     0 &     0 &     0 &     0 &     0 &     0 &     0 &   -x_9 &   -x_8 &     0 &     0 &     0 &    x_6 &     0 &    x_7 &     0 &    x_4 & 0 &   -x_5 &     0 &   -x_2 &     0 &    x_3 &     0 &    x_0 &     0 &   -x_1\\
   0 &     0 &     0 &     0 &     0 &     0 &     0 &     0 &     0 &     0 &     0 &     0 &     0 &     0 &   -x_8 &    x_9 &    x_7 &     0 &   -x_6 &     0 &    x_5 &     0 &    x_4 &     0 &   -x_3 &     0 &   -x_2 &     0 &    x_1 &     0 &    x_0 &     0\\
   0 &     0 &     0 &     0 &     0 &     0 &     0 &     0 &     0 &     0 &     0 &     0 &     0 &     0 &   -x_9 &   -x_8 &     0 &    x_7 &     0 &   -x_6 &     0 &    x_5 & 0 &    x_4 &     0 &   -x_3 &     0 &   -x_2 &     0 &    x_1 &     0 &    x_0\\
\end{array}\right]
\end{eqnarray}
}
\hrule
\end{figure*}

Adams-Lax-Phillips has also provided another construction of square RODs using Quaternions~\cite{Lia}. 
Assuming that a square ROD of order $n=2^a$
\[
\mathbb{O}_n^{(Q)}=\mathbb{O}_n^{(Q)}(x_0,\cdots,x_{\rho(n)-1})
\]
which has $\rho(n)$ real variables, is given, then a square ROD of order $16n$ with $\rho(n)+8$
real variables $x_i$ for  $i=0,1,\cdots, \rho(n)+7$
\[
\mathbb{O}_{16n}^{(Q)}=\mathbb{O}_{16n}^{(Q)}(x_0,\cdots,x_{\rho(n)+7})
\]
is given by \eqref{ALPQ}, %as shown at the top of this page 
where the matrices $L_4$ and $R_4$ are given by
\begin{eqnarray*}
L_4(x_0,x_1,x_2,x_3)=\left(\begin{array}{rrrr}
    x_0   & x_1   & x_2    & x_3 \\
   -x_1   & x_0   &-x_3    & x_2 \\
   -x_2   & x_3   & x_0    &-x_1 \\
   -x_3   &-x_2   & x_1    & x_0
\end{array}\right),\\
R_4(x_4,x_5,x_6,x_7)=\left(\begin{array}{rrrr}
    x_4   & x_5   & x_6    & x_7 \\
   -x_5   & x_4   & x_7    &-x_6 \\
   -x_6   &-x_7   & x_4    & x_5 \\
   -x_7   & x_6   &-x_5    & x_4
\end{array}\right).
\end{eqnarray*}
respectively with $y_i=x_{8+i}$.\\
%%%%%%%%%%%%%
The Geramita-Pullman construction of square RODs~\cite{Lia} is given as follows.

Consider a recursive construction of square ROD of order $n=2^a$ to $16n$ as follows:
$\mathbb{O}_n^{(GP)}=\mathbb{O}_n^{(GP)}(x_0,\cdots,x_{\rho(n)-1})$
which has $\rho(n)$ real variables is given, then a square ROD $\mathbb{O}_{16n}^{(GP)}$ of order $16n$ with $\rho(n)+8$ real variables $x_i$ for  $i=0,1,\cdots, \rho(n)+7$ is given by

{\tiny{
\begin{eqnarray}
\label{geramita}
 \left[\begin{array}{cc}
%{l@{\hspace{0.4pt}}l@{\hspace{0.4pt}}}
  K_8(x_0,\cdots,x_7)\otimes I_n & I_8 \otimes\mathbb{O}_n^{(GP)}(y_0,\cdots,y_{\rho(n)-1})\\
I_8\otimes (-\mathbb{O}_n^{(GP)})^\mathcal{T}(y_0,\cdots,y_{\rho(n)-1}) & K_8^\mathcal{T}(x_0,\cdots,x_7)\otimes I_n\\
\end{array}\right]
 \end{eqnarray}
}}

\noindent
with $y_i=x_{8+i}$.

It can be checked that both Adams-Lax-Phillips construction from Quaternions and Geramita-Pullman's construction differ from the constructions of $\mathbb{O}_{16n}^{(Q)}$ and $\mathbb{O}_{16n}^{(GP)}$ defined above only in rearrangement of variables and in signs of some of the rows or columns of the design matrix.
%%%%%%%%%%%%%%
\begin{lemma}
\label{AtQGP}
Let $t \geq 16$ and $\hat{A}_t$ and $P_t$ be the square RODs of order $t$ given by Lemma \ref{threedesign} (ii) and (iii) respectively, and also let $\mathbb{O}_{16n}^{(Q)}$ and $\mathbb{O}_{16n}^{(GP)}$ be the square RODs of order $16n$ given by \eqref{ALPQ} and \eqref{geramita} respectively. Then $\hat{A}_t=\mathbb{O}_{16n}^{(Q)}$ and $P_t=\mathbb{O}_{16n}^{(GP)}$ for $t=16n$   .
\end{lemma}
%%%%%%%%%%
\begin{proof}
Similar to that of Lemma \ref{AtO2} and hence omitted.
\end{proof}

%%%%%%%%%%%%%
Note that the square RODs for more than than $8$ antennas obtained by Adams-Lax-Phillips construction from Octonion and Quaternion are different from the square RODs constructed in this paper (denoted by $R_t$, $t$ a power of $2$).
On the other hand, the square ROD $P_{16}$ for $16$ antennas obtained by Geramita-Pullman construction is exactly the square ROD $R_{16}$ given by \eqref{R16}. However, for more than $16$ antennas, they are not identical. For example, the ROD
%%%%%%%%%%
$P_{32}$ of size $[32,32,10]$ (given by \eqref{GP32}) is different from the matrix $R_{32}$ given by \eqref{R32}.

\section{rate-1/2 scaled-COD of size $[32,10,16]$}
\label{appendixIII}
\vspace{-10pt}
{\scriptsize
\begin{eqnarray*}
\vspace{-40pt}
\hspace{25pt}
 \left[
\begin{array}{ r @{\hspace{.2pt}} r @{\hspace{.2pt}} r @{\hspace{.2pt}} r @{\hspace{.2pt}} r @{\hspace{.2pt}} r @{\hspace{.2pt}} r @{\hspace{.2pt}} r @{\hspace{.2pt}} r @{\hspace{.2pt}} r}
     x_0 & -x_1^* & -x_2^* & 0   & -x_3^* & 0 & 0 & 0 &-\frac{x_7^*}{\sqrt{2}} & -\frac{x_{15}^*}{\sqrt{2}}\\
     x_1 & x_0^*  & 0  & -x_2^* & 0 & -x_3^* & 0 & 0  & \frac{x_6^*}{\sqrt{2}} & \frac{x_{14}^*}{\sqrt{2}}\\
     x_2 & 0 & x_0^* & x_1^*  & 0 & 0 & -x_3^* & 0 &-\frac{x_5^*}{\sqrt{2}} & -\frac{x_{13}^*}{\sqrt{2}}\\
     0   & x_2 & -x_1 & x_0  & 0 & 0 & 0 & -x_3^*  &-\frac{x_4}{\sqrt{2}} & -\frac{x_{12}}{\sqrt{2}}\\
     x_3 & 0 & 0 & 0 & x_0^* & x_1^* &x_2^* &0    & \frac{x_4^*}{\sqrt{2}}  & \frac{x_{12}^*}{\sqrt{2}} \\
     0   & x_3 & 0 & 0  & -x_1 & x_0 &0 &x_2^*     &-\frac{x_5}{\sqrt{2}} & -\frac{x_{13}}{\sqrt{2}}\\
     0   &0  & x_3 & 0  & -x_2 &0 & x_0 & -x_1^* &-\frac{x_6}{\sqrt{2}} & -\frac{x_{14}}{\sqrt{2}}\\
     0   &0  &0 & x_3 & 0  & -x_2  & x_1 & x_0^* &-\frac{x_7}{\sqrt{2}} & -\frac{x_{15}}{\sqrt{2}}\\
     x_8  & -x_9^* & -x_{10}^* & 0   & -x_{11}^* & 0 & 0 & 0 &-\frac{x_{15}^*}{\sqrt{2}} & \frac{x_7^*}{\sqrt{2}}\\
     x_9 & x_8^*  & 0  & -x_{10}^* & 0 & -x_{11}^* & 0 & 0  & \frac{x_{14}^*}{\sqrt{2}} &-\frac{x_6^*}{\sqrt{2}}\\
     x_{10} & 0 & x_8^* & x_9^*  & 0 & 0 & -x_{11}^* & 0 &-\frac{x_{13}^*}{\sqrt{2}} & \frac{x_5^*}{\sqrt{2}}\\
     0  & x_{10} & -x_9 & x_8  & 0 & 0 & 0 & -x_{11}^*  &-\frac{x_{12}}{\sqrt{2}} & \frac{x_4}{\sqrt{2}}\\
     x_{11} & 0 & 0 & 0 & x_8^* & x_9^* &x_{10}^* &0    & \frac{x_{12}^*}{\sqrt{2}} &-\frac{x_4^*}{\sqrt{2}}\\
     0  & x_{11} & 0 & 0  & -x_9 & x_8 &0 &x_{10}^*     &-\frac{x_{13}}{\sqrt{2}} & \frac{x_5}{\sqrt{2}} \\
     0  &0  & x_{11} & 0  & -x_{10} &0 & x_8 & -x_9^* &-\frac{x_{14}}{\sqrt{2}} & \frac{x_6}{\sqrt{2}}\\
     0  &0  &0 & x_{11} & 0  & -x_{10}  & x_9 & x_8^* &-\frac{x_{15}}{\sqrt{2}} & \frac{x_7}{\sqrt{2}} \\
    x_4 &-x_5^* &-x_6^* &-x_7^*  & 0 & 0 & 0 & 0  &-\frac{x_3^*}{\sqrt{2}} & \frac{x_{11}^*}{\sqrt{2}}\\
     x_5& x_4^* & 0 & 0  & -x_6^* &-x_7^* & 0 & 0   &\frac{x_2^*}{\sqrt{2}} &-\frac{x_{10}^*}{\sqrt{2}}\\
    x_6 &0  & x_4^* & 0  & x_5^* &0 &-x_7^* & 0 &-\frac{x_1^*}{\sqrt{2}}  & \frac{x_9^*}{\sqrt{2}}\\
     0    & x_6 & -x_5 &0  & x_4  & 0  & 0 &-x_7^* &-\frac{x_0}{\sqrt{2}} & \frac{x_8}{\sqrt{2}}\\
     x_7 &0 &0 & x_4^*  &0 & x_5^*   &-x_7^*  & 0  & \frac{x_0^*}{\sqrt{2}} &-\frac{x_8^*}{\sqrt{2}}\\
     0    & x_7  & 0  & -x_5 & 0 & x_4 & 0 & x_6  &-\frac{x_1}{\sqrt{2}}  & \frac{x_{10}}{\sqrt{2}}\\
     0    & 0 & x_7 &-x_6  & 0 & 0 & x_4 &-x_5^* &-\frac{x_2}{\sqrt{2}}   & \frac{x_{10}}{\sqrt{2}}\\
     0    & 0 & 0 & 0 & x_7 &-x_6 & x_5   & x_4^*  &-\frac{x_3}{\sqrt{2}}  & \frac{x_{11}}{\sqrt{2}} \\
     x_{12} &-x_{13}^* &-x_{14}^* &-x_{15}^*  & 0 & 0 & 0 & 0  &-\frac{x_{11}^*}{\sqrt{2}} &-\frac{x_3^*}{\sqrt{2}}\\
     x_{13}& x_{12}^* & 0 & 0  & -x_{14}^* &-x_{15}^* & 0 & 0  & \frac{x_{10}^*}{\sqrt{2}} & \frac{x_2^*}{\sqrt{2}}\\
    x_{14} &0  & x_{12}^* & 0  & x_{13}^* &0 &-x_{15}^* & 0 &-\frac{x_9^*}{\sqrt{2}} &-\frac{x_1^*}{\sqrt{2}}\\
     0    & x_{14} & -x_{13} &0  & x_{12}  & 0  & 0 &-x_{15}^* &-\frac{x_8}{\sqrt{2}} &-\frac{x_0}{\sqrt{2}}\\
     x_{15} &0 &0 & x_{12}^*  &0 & x_{13}^*   &-x_{15}^*  & 0  & \frac{x_8^*}{\sqrt{2}} & \frac{x_0^*}{\sqrt{2}}\\
     0    & x_{15}  & 0  & -x_{13} & 0 & x_{12} & 0 & x_{14}  &-\frac{x_9}{\sqrt{2}} &-\frac{x_1}{\sqrt{2}}\\
     0    & 0 & x_{15} &-x_{14}  & 0 & 0 & x_{12} &-x_{13}^* &-\frac{x_{10}}{\sqrt{2}}   &-\frac{x_2}{\sqrt{2}}\\
     0    & 0 & 0 & 0 & x_{15} &-x_{14} & x_{13}   & x_{12}^*  &-\frac{x_{11}}{\sqrt{2}}  &-\frac{x_3}{\sqrt{2}}
\end{array}\right]
\end{eqnarray*}
 }
\end{appendices}

\newpage
%%%%%%%%%%%%%%%%%%%%%%%%%%%%%%%%%%%%%%%%%%%%%%%%%%%%%%%%%%%%%%%%%%%%%%%%%%%
\onecolumn

%\vspace*{-35\baselineskip}

\begin{IEEEbiographynophoto}{Smarajit Das}

(S'2007-M'2010) was born in West Bengal, India. He received his B.E. degree from the Sardar Vallabhbhai National Institute of Technology, Surat, India,  M.Tech in electrical engineering from the Indian Institute of Technology, Delhi, India,  and the Ph.D. degree in electrical communication engineering from the Indian Institute of Science, Bangalore, India, in 2001, 2003 and 2009 respectively.
He is currently a Post-doctoral fellow in the School of Technology and Computer science, Tata Institute of Fundamental Research, Mumbai, India. His primary research interests are in algebraic coding, Classical and Quantum Information theory.

\end{IEEEbiographynophoto}

%\vspace*{-30\baselineskip}

\begin{IEEEbiographynophoto}{B.~Sundar~Rajan}
(S'84-M'91-SM'98) was born in Tamil Nadu, India. He received the B.Sc. degree in mathematics from Madras University, Madras, India, the B.Tech degree in electronics from Madras Institute of Technology, Madras, and the M.Tech and Ph.D. degrees in electrical engineering from the Indian Institute of Technology, Kanpur, in 1979, 1982, 1984, and 1989 respectively. He was a faculty member with the Department of Electrical Engineering at the Indian Institute of Technology in Delhi, from 1990 to 1997. Since 1998, he has been a Professor in the Department of Electrical Communication Engineering at the Indian Institute of Science, Bangalore. His primary research interests include space-time coding for MIMO channels, distributed space-time coding and cooperative communication, modulation and coding for multiple-access and relay channels, and network coding. 

Dr. Rajan is an Editor of the IEEE \textsc{Transactions on Wireless Communications}, an Editor of \textsc{IEEE Wireless Communications Letters}, and an Editorial Board Member of \emph{International Journal of Information and Coding Theory}. He was an Associate Editor of the \textsc{IEEE Transactions on Information Theory} during 2008-2011. He served as Technical Program Co-Chair of the IEEE Information Theory Workshop (ITW'02), held in Bangalore, in 2002. He is a Fellow of the Indian National Science Academy, a Fellow of the Indian National Academy of Engineering, and a Fellow of the National Academy of Sciences, India.  He is a recipient of IEEE Wireless Communications and Networking Conference 2011 Best Academic Paper Award,  a  recipient of Prof. Rustum Choksi award by I.I.Sc., for excellence in research in Engineering for the year 2009, recipient of the Khosla National Award from I.I.T. Roorkee for the year 2010, and recipient of the IETE Pune Center's S.V.C Aiya Award for Telecom Education in 2004. Dr. Rajan is a Member of the American Mathematical Society.

\end{IEEEbiographynophoto}

\vfill


\begin{thebibliography}{160}

\bibitem{Lia} X. B. Liang
``Orthogonal Designs with Maximal Rates,''
{\it{IEEE Trans. Inform. Theory,}} Vol. 49, no. 10, pp. 2468-2503, Oct. 2003

\bibitem{TaK}
Vahid Tarokh and Il-Min Kim, ``Existence and construction of noncoherent unitary space-time codes,'' {\em IEEE Trans. Inform. Theory}, vol. 48, no. 12, pp. 3112-3117, Dec. 2002.

\bibitem{JSKR}
R.~V.~J.~R.~Doddi, V.~Shashidhar, Md. Zafar Ali Khan and B.~Sundar Rajan, ``Low-complexity, full-diversity, space-time frequency block codes for MIMO-OFDM,'' {\em Proc. IEEE GLOBECOM 2004},  Dallas, Texas, Nov. 29-Dec. 3, pp. 204-208, 2004.

\bibitem{JiH}
Yindi Jing and B. Hassibi,
``Distributed space-time coding in wireless relay networks,'' {\em IEEE Trans. Wireless Communications,} vol. 5, no. 12, pp. 3524-3536, Dec. 2006.

% \bibitem{ChS}

\bibitem{TiH} O. Tirkkonen and  A. Hottinen,
``Square matrix embeddable STBC for complex signal constellations Space-time block codes from orthogonal design,''
{\it{IEEE Trans. Inform. Theory,}} Vol. 48, no. 2, pp. 384-395, Feb. 2002.

\bibitem{TJC} V. Tarokh, H. Jafarkhani, and A. R. Calderbank,
``Space-time block codes from orthogonal designs,''
{\it{IEEE Trans. Inform. Theory,}} vol. 45, no. 5, pp. 1456-1467, July 1999.

\bibitem{ALP} J. F. Adams, P. D. Lax, and R. S. Phillips, ``On matrices whose real linear combinations are nonsingular,''
{\it{ Proc. Amer. Math. Soc.,}} vol. 16, no. 2, pp. 318-322, Apr. 1965.

\bibitem{LFX} Kejie Lu, Shengli Fu and Xiang-G Xia,
``Closed-Form Designs of Complex Orthogonal Space-Time Block Codes of Rates $\frac{k+1}{2k}$ for $2k-1$ or $2k$ Transmit Antennas,''
{\it{IEEE Trans. Inform. Theory,}} vol. 51, No. 5, pp. 4340-4347, Dec. 2005.

\bibitem{AKP} Sarah Spence Adams, Nathaniel Karst and Jonathan Pollak,
``The Minimum Decoding Delay of Maximal Rate Complex Orthogonal Space-Time Block Codes,''
{\it{IEEE Trans. Inform. Theory,}} vol. 53, No. 8, pp. 2677-2684, Aug. 2007.

\bibitem{AKM} Sarah Spence Adams, Nathaniel Karst and Mathav Kishore Murugan,
``The Final Case of the Decoding Delay Problem for Maximum Rate Complex Orthogonal Designs,''
{\it{IEEE Trans. Inform. Theory,}} vol. 56, No. 1, pp. 103-112, Jan. 2010.


 \bibitem{GeP} A. V. Geramita and N. J. Pullman, ``A theorem of  Hurwitz and Radon and Orthogonal projective modules,''
{\it{ Proc. Amer. Math. Soc.,}} vol. 42, No. 1, pp. 51-56, Jan. 1974.

\bibitem{DaR} Smarajit Das and B. Sundar Rajan,
``Square Complex Orthogonal Designs with Low PAPR and Signaling Complexity,''
{\it{IEEE Trans. Wireless Communications,}} Vol. 8, No. 1, pp. 204-213, Jan. 2009.


\bibitem{WaX} Haiquan Wang and Xiang-Gen Xia,
``Upper Bounds of Rates of Complex Orthogonal Space-Time Block Codes,''
{\it{IEEE Trans. Inform. Theory,}} vol. 49,  No. 10, pp. 2788-2796, Oct. 2003.

\bibitem{Sha} D. B. Shapiro, Compositions of Quadratic forms, Berlin, Germany: Walter de Gruyter, 2000.


\bibitem{TWSMS} L. C. Tran, T. A. Wysocki, J. Seberry, A. Mertins, and S. A. Spence,
``Generalized Williamson and Wallis-Whiteman constructions for improved
square order-8 CO STBCs,'' {\it{ Proc. IEEE PIMRC,}} 11-14 Sep., pp. 1155-1159, 2005.

\bibitem{SSW} J. Seberry, S. A. Spence, and T. A. Wysocki, ``A construction technique
for generalized complex orthogonal designs and applications to
wireless communications,''{\it{ Linear Algebra Appl.,}} vol. 405, pp. 163-176, Aug. 2005.



\end{thebibliography}
\end{document}